	\documentclass[11pt]{article} 
	\usepackage{amsmath} % AMS Math Package
	\usepackage{amsthm} % Theorem Formatting
	\usepackage{amssymb}	% Math symbols such as \mathbb
	\usepackage{graphicx} % Allows for eps images
	\usepackage{multicol} % Allows for multiple columns
	\usepackage[dvips,letterpaper,margin=1 in,bottom=0.5in]{geometry}
	\usepackage{color}
	\usepackage{setspace}
	\usepackage{indentfirst} 
	\usepackage{textcomp}
	\usepackage{empheq}
	\makeatletter % Need for anything that contains an @ command 
	\renewcommand{\maketitle} % Redefine maketitle to conserve space
	{ \begingroup \vskip 10pt \begin{center} \large {\bf \@title}
		\vskip 10pt \large \@author \hskip 20pt \@date \end{center}
	  \vskip 10pt \endgroup \setcounter{footnote}{0} }
	\makeatother % End of region containing @ commands
	\newcommand{\be }{\begin{equation}}
	\newcommand{\ee}{\end{equation}}
	 \newcommand{\vpa}{\vec{\partial}}%partial with overhead arrow
	 
	 \newcommand{\gld}[1]{\widehat{\mathcal{L}}_{#1}}
	 \newcommand{\lb}{\left}
	 \newcommand{\rb}{\right}
	 \newcommand{\bigO}{\mathcal{O}}
	\newcommand{\SUL}[2]{\overset{#2}{\underset{#1}{\sum}}}%sum with upper and lower limit
	\let\baraccent=\= % rename builtin command \= to \baraccent
	\renewcommand{\=}[1]{\stackrel{#1}{=}} % for putting numbers above =
	
	\newtheorem{thm}{Theorem}[section]
	\newtheorem*{conj}{HZ conjecture}
	
	\theoremstyle{definition}
	
	\theoremstyle{remark}

	\numberwithin{equation}{section} % Number equations within sections (i.e. 1.1, 1.2, 2.1, 2.2 instead of 1, 2, 3, 4)
	\numberwithin{figure}{section} % Number figures within sections (i.e. 1.1, 1.2, 2.1, 2.2 instead of 1, 2, 3, 4)
	\numberwithin{table}{section} % Number tables within sections (i.e. 1.1, 1.2, 2.1, 2.2 instead of 1, 2, 3, 4)
	
	\usepackage{jheppub} % for details on the use of the package, please
	\usepackage[T1]{fontenc} % if needed

	\title{\boldmath  A note on large gauge transformations in double field theory}
	\author{Usman Naseer}
	\affiliation{Center for Theoretical Physics,\\
	Massachusetts Institute of Technology,\\
	Cambridge, MA 02139, USA.
	}
	\emailAdd{unaseer@mit.edu}
	
	\abstract{
	We give a detailed proof of the conjecture by Hohm and Zwiebach in double field theory. This result implies that their proposal for large gauge transformations in terms of the Jacobian matrix for coordinate transformations is, as required, equivalent to the standard exponential map associated with the generalized Lie derivative along a suitable parameter.
	}
	
\begin{document}
	\rightline{}
	\rightline{\tt }
	\rightline{\tt  MIT-CTP/4665}
    \rightline{April 2015}
	\onehalfspacing
	
	\maketitle
	
	\section{Introduction}
	Double field theory was developed to make manifest the $O(D,D)$ T-duality  symmetry in the low energy effective field theory limit of string theory. In addition to the usual spacetime coordinates, `winding' coordinates are introduced. The metric and the Kalb-Ramond two-form are combined into a `generalized metric'. This generalized metric transforms linearly under global $O(D,D)$ transformations. Gauge transformations of the fields can also be written in an $O(D,D)$ covariant form and they can be interpreted as the `generalized coordinate transformations' in the doubled spacetime as discussed in \cite{LargeGT}. The action of double field theory, written in terms of the generalized metric, is then manifestly invariant under these transformations. Double field theory is a restricted theory. The so-called strong constraint restricts the theory to live on a $D$-dimensional subspace of the doubled spacetime. Different solutions of the strong constraint are then related by T-duality. Double field theory was developed in \cite{Siegel,Siegel2,hullzw1,BGIaction,GaugeAlgebra,GenMetricForm} and earlier ideas can be found in \cite{DSform,DSCST,DRinST,DRot}. Further developments of double field theory are discussed in \cite{FlikeGofDFT,KwakEOM,TvsG,Lnotes,HString,UniII,DFTII,MassiveII,SUSYDFT,RTensor,GGandM,DGP2,BMP,BGGP,GSTIIA,RW,DGAPP,StringyDG,FermionDFT,SDFT,RRCoh,Tfolds,Connecting,sigma,MtoD,DoubleBranes,EffAc,N4gauged,Gauged,SUGRAasGG,GGMth,GofDFT}. For recent reviews on this subject see \cite{Rev2} and \cite{Rev1}.
	
	\par The important issue of large gauge transformations in double field theory was discussed in \cite{LargeGT} where a formula for finite gauge transformations of fields has been proposed. These transformations are induced by generalized coordinate transformations. Under the transformation from cordinates $X^M$ to coordinates $X'^M$ fields are claimed to transform by the matrix $\mathcal{F}$ given by
	\begin{eqnarray}
		\mathcal{F}_M^{\;\;N}=\frac{1}{2}\left(\frac{\partial X^P }{\partial X^{\prime M}}\frac{\partial X_{P}^{\prime}}{\partial X_N}+\frac{\partial X_M^{\prime} }{\partial X_P}\frac{\partial X^N}{\partial X^{\prime P}}\right).\label{defF}
	\end{eqnarray}
	A generalized vector $V_M$ transforms as
	\be 
	V^{\prime}_{M}(X^{\prime})=\mathcal{F}_{M}^{\;\;N}V_{N}(X), \label{VecRule}
	\ee
	or, in index-free notation,
	\be 
	V^{\prime}(X^{\prime})=\mathcal{F}V(X),\label{compII}
	\ee
	and fields with more than one $O(D,D)$ index transform tensorially where each index is transformed by $\mathcal{F}_{M}^{\;\;N}$. If the finite transformation of the coordinates is generated by a parameter $\xi^M(X)$, i.e.,
	\begin{eqnarray}
		X^{\prime ^M}=e^{-\xi^P\partial_P}X^M, \label{Xprime} 
	\end{eqnarray}
	then $\mathcal{F}$ can be expressed in terms of this new parameter, i.e. $\mathcal{F}=\mathcal{F}(\xi)$. 
	\par By considering an infinitesimal coordinate transformation $X^{\prime ^M}=X^M-\xi^M$ \cite{LargeGT}, the transformation rule (\ref{VecRule}) leads to the infinitesimal transformation of a generalized vector, which is given by the generalized Lie derivative $\gld{\xi}$ as follows:
	\begin{eqnarray}
		V^{\prime}_M(X)&=&V_M(X)+\xi^K\partial_KV_M(X)+\lb(\partial_M\xi^K-\partial^K\xi_M\rb)V_K(X)
		=V_M(X)+\lb(\hat{\mathcal{L}}_{\xi}V(X)\rb)_M.\ \ \ 
	\end{eqnarray}Generalized Lie derivatives define a Lie algebra \cite{GenMetricForm} but the closure of the Lie algebra requires the strong constraint
	\be \partial^P\partial_P\lb(\cdots\rb)=0,\label{strongconst}\ee
	where `$\cdots$' indicate arbitrary fields, parameters or their products (so that $\partial^P\partial_P\ A =0\ $ and $\partial^PA\ \partial_P\ B=0\ $ for any fields or parameters $A$ and $B$). We can also realize large gauge transformations by exponentiating the generalized Lie derivative.  Then, as discussed in \cite{LargeGT}, these transformations will form a group. For a vector, they are given by
	\begin{eqnarray}
		V'(X')&=& e^{-\xi^P\partial_P}\ e^{\hat{\mathcal{L}}_{\xi}}\ V(X)\ \equiv\ \mathcal{G}\lb(\xi\rb)\  V(X),\label{compI}
	\end{eqnarray}
	where we make a useful definition
	\begin{eqnarray}
		\mathcal{G}\lb(\xi\rb)\equiv\ e^{-\xi^P\partial_P}\ e^{\hat{\mathcal{L}}_{\xi}}.
		\label{defG}
	\end{eqnarray}
	From (\ref{compI}) and (\ref{compII}) we would expect that 
	\be  
	\mathcal{F}(\xi)\ \stackrel{?}{=}\ \mathcal{G}(\xi)\label{conj1},
	\ee but this is not true and this fact does not lead to any inconsistency. The important question is whether there exists a generalized coordinate transformation for which the corresponding field transformation given by (\ref{defF}) equals $\mathcal{G}\lb(\xi\rb)$. This is equivalent to asking whether it is possible to find a parameter $\xi'(\xi)$ such that $\ \mathcal{F}\lb(\xi\rb)=G\lb(\xi'\lb(\xi\rb)\rb)$. This is precisely the question we address in this paper and the answer is positive. Consistency requirements lead to the condition that $\xi'$ and $\xi$ can only differ by a quasi-trivial parameter 
	\be 
	\xi'^P(\xi)=\xi^P+\sum_{i}\rho_{i}(\xi)\partial^P\eta_{i}(\xi),
	\ee
	where we have defined a quasi-trivial parameter to have the form $\sum_{i}\rho_{i}(\xi)\partial^P\eta_{i}(\xi)$, where the free index is carried by a derivative. We will also refer to the parameters of the form $\partial^M f$ as `trivial' parameters. Note that every trivial parameter is also quasi-trivial. It is also easy to see, using the strong constraint, that $\gld{\xi}$ is zero for a trivial parameter and it is not zero in general for a quasi-trivial parameter.
	
	\par These insights were used in \cite{LargeGT} to determine such $\xi'(\xi)$ to cubic order in $\xi$. We also note a small difference in notation and approach as compared to \cite{LargeGT} where parameters $\xi$ and $\Theta\lb(\xi\rb)$ are related by $\ \mathcal{G}\lb(\xi\rb)=\mathcal{F}\lb(\Theta\lb(\xi\rb)\rb)\ $. Since $\Theta\lb(\xi\rb)$ is just a sum of $\xi$ and higher order terms in $\xi$, this relation can be inverted to obtain $\xi\lb(\Theta\rb)$. With this in mind, the conjecture by Hohm and Zwiebach (HZ conjecture) can be stated as
	\begin{conj}For every parameter $\xi\lb(X\rb)$ there exists a parameter $\xi^{\prime}(\xi)$  such that
		$\ \ 
		\xi^{\prime M}\lb(\xi\rb)=\xi^M+\delta^M(\xi)\ 
		$
		where  $\delta^M(\xi)$  is quasi-trivial and
		\be  \mathcal{G}\lb(\xi^{\prime}\lb(\xi\rb)\rb)=\mathcal{F}(\xi).\label{claim}\ee
			
	\end{conj}
	The main aim of this paper is to prove this conjecture by giving a procedure which can be used to obtain $\xi'(\xi)$ to all orders in $\xi$. We will do this by finding a parameter $\xi'(\xi)$ such that  $\mathcal{G}(\xi^{\prime}(\lambda\xi))$ and $\mathcal{F}(\lambda\xi)$ satisfy the same first order differential equation in $\lambda$ and the same initial condition at $\lambda=0$ . By the uniqueness of the solution we can then deduce that \be  \mathcal{F}(\lambda\xi)=\mathcal{G}(\xi'\lb(\lambda\xi\rb)),\ee
	and by setting $\lambda=1,$ we get (\ref{claim}).
	 
	The HZ conjecture was examined explicitly in \cite{GlobalAspects}; see also \cite {Park} and \cite{Hull} for other relevant discussion. In \cite{GlobalAspects}, Berman, Perry and Cederwall considered this issue in the context of equivalence classes of transformations. These equivalence classes were defined modulo the so-called `non-translating' transformations, which are the transformations generated by the quasi-trivial parameters. These transformations differ from the identity by nilpotent matrices. Important progress was made by showing that $\mathcal{F}\lb(\xi\rb)$ and $\mathcal{G}\lb(\xi\rb)$ belong to the same equivalence class but the issue of the existence of $\xi'\lb(\xi\rb)$ such that (\ref{claim}) holds was not addressed completely. Our results in section (\ref{decG}) play a key role in extending the results of \cite{GlobalAspects} to prove HZ conjecture. Our construction here is more straightforward and one can in fact show, without finding the explicit form of $\xi'\lb(\xi\rb)$, that there exists a parameter $\xi'\lb(\xi\rb)$ such that (\ref{claim}) holds.
	
	\par This paper is organized as follows. In section \ref{rev}, we review the basics of gauge transformations in double field theory. We also review how the consistency requirements for the transformations of the scalar field constrain the form of $\xi'\lb(\xi\rb)$. In section \ref{decG}, we show that  $\mathcal{G}\lb(\xi'\lb(\xi\rb)\rb)$ can be written as the product $\mathcal{G}\lb(\xi\rb)\mathcal{G}\lb(\chi\lb(\xi\rb)\rb)$, where $\chi\lb(\xi\rb)$ is a quasi-trivial parameter. Further we show how $\chi\lb(\xi\rb)$ can be used to find $\xi'\lb(\xi\rb)$. Differential equations for $\mathcal{F}\lb(\lambda\xi\rb)$ and $\mathcal{G}\lb(\xi'\lb(\lambda\xi\rb)\rb)$ are derived in sections \ref{difF} and \ref{deqG} respectively. By comparing the two differential equations, we obtain an iterative relation (\ref{Iterative}). In section \ref{Method}, we give a procedure to use this relation  systematically to obtain $\xi'\lb(\xi\rb)$ to all orders in $\xi$, proving the HZ conjecture. We use our procedure to compute $\xi'\lb(\xi\rb)$ up to quintic order in section \ref{checks}. Up to quartic order, the results are found to be in agreement with \cite{LargeGT}. For the quintic order, the result is checked by explicit computation using Mathematica. In section \ref{Allord}, we give an explicit formula for $\chi\lb(\xi\rb)$ to all orders in $\xi$. We also discuss the connection between our results and those in \cite{GlobalAspects} in section \ref{Allord}. In section \ref{concl}, we discuss the implications of our results on the composition of coordinate transformations in double field theory. Finally, we comment on finite gauge transformations in exceptional field theory. 
	\section{Preliminaries}\label{prelem}
	\subsection{Review}\label{rev}
	In this section we review the basics of gauge transformations in DFT and introduce notation and conventions required to derive the relevant differential equations. Most of these relations were first derived in \cite{LargeGT} and we present them here without proof. 
	
	Consider a finite coordinate transformation which is generated by a parameter $\xi^M$, i.e.,
	\be \label{zetaDef}
	X'^M\ = \ X^M-\ \zeta^M \lb(\xi\rb) \ \equiv \ \lb(e^{-\xi^P\ \partial_P}\rb)\  X^{M},
	\ee
	or in a more transparent, index free notation,
	\be 
	X'\ =\ X-\ \zeta(\xi)= e^{-\xi^P\partial_P}\ X\ \ .\label{gctrans}
	\ee
	Using the definition of $\mathcal{F}$ as in (\ref{defF}), we can express it in terms of the generating parameter $\xi^M$ as follows.
	\begin{eqnarray}
		\mathcal{F}(\xi)&=&\frac{1}{2}(e^{-\xi}e^{\xi+a}e^{-\xi}e^{\xi-a^t}+e^{-\xi}e^{\xi-a^t}e^{-\xi}e^{\xi+a}).\label{Fexi}   
	\end{eqnarray}
	Here $a_M^{\;\;N}\ \equiv\ \partial_M\xi^N$ and $\xi$ appearing on the RHS is to be understood as the differential operator $\xi \cdot f \ \equiv\  \xi^P\partial_P\ f$. This expression for $\mathcal{F}$ may look like a differential operator but it is a matrix function as demonstrated in section 4.1 of \cite{LargeGT}.
	The proof essentially follows from the fact that $\  e^{-\xi}\ e^{\xi+a}\ $ is not a differential operator and it can be written in the following form, which makes its matrix nature manifest.
	\be  e^{-\xi}e^{\xi+a}=\lb(\mathbf{1}\cdot e^{-\overleftarrow{\xi}+a}\rb),\ee
	where we have introduced the notation
	\be \beta\overleftarrow{\xi}\Lambda\equiv \lb(\xi^P\partial_P\beta\rb) \Lambda.\ee
	Note that we can also write $\mathcal{G}(\xi)$ in a similar manner as
	\be 
	\mathcal{G}(\xi)= e^{-\xi}e^{\gld{\xi}}=e^{-\xi}e^{\xi+a-a^t}=\lb(\mathbf{1}\cdot e^{-\overleftarrow{\xi}+a-a^t}\label{diffformG}\rb).
	\ee
	We make the useful definition
	\begin{eqnarray}
		A_{M}^{\ N} (\xi) &\equiv& \partial_M \zeta^N\lb(\xi\rb)\ =\  -\SUL{n=1}{\infty}\frac{(-1)^n}{n!}\partial_M\ \lb(\left(\xi^P\ \partial_P\right)^{n-1}\ \xi^N\rb),\label{defA}
	\end{eqnarray}
	where $\zeta^N\lb(\xi\rb)$ is read from equation (\ref{zetaDef}). Further we make the definition
	\begin{eqnarray}
		B_{M}^{\ N}(\xi) &\equiv&\  \frac{\partial X^N}{\partial X'^M},
	\end{eqnarray}
	which implies that
	\begin{eqnarray}
		B(\xi) &=&\  \frac{1}{1-A(\xi)}\ =\ e^{-\xi}e^{\xi+a} =\lb(\mathbf{1} \cdot e^{-\overleftarrow{\xi}+a}\rb). \label{defM}
	\end{eqnarray}
	The equalities in the last line follow from equation (\ref{gctrans}) and the definition of $A$ as in equation (\ref{defA}). Details can be found in section 4.1 of \cite{LargeGT}.
	Using these definitions,  $\mathcal{F(\xi)}$ can be written in terms of the matrix $B(\xi)$. We suppress the explicit $\xi$ dependence here and write
	\begin{eqnarray}
		\mathcal{F} = \frac{1}{2}\lb(B\ \lb(B^{-1}\rb)^t+\lb(B^{-1}\rb)^t\ B\rb).\label{FasB}
	\end{eqnarray}
	
	Now we turn our attention to the finite transformations realized as the exponentiation of the generalized Lie derivative. As mentioned previously, the infinitesimal transformation of a vector field $V(X)$ is given by the action of the generalized Lie derivative, i.e.,
	\begin{eqnarray}
		V'(X)\ &=&\ V(X) + \widehat{\mathcal{L}}_{\xi} V(X),\end{eqnarray}
		with 
		\begin{eqnarray}\ \gld{\xi} V_{M}(X) \ &=&\ \xi^K\partial_KV_M(X)+\lb(\partial_M\xi^K-\partial^K\xi_M\rb)V_K(X).
		\end{eqnarray}
		or in the index free notation,
		\be 
		\gld{\xi} V(X) =\xi\cdot V(X) +\left(a\ -\ a^t\right) V(X) \ \  \text{with }\ a_{M}^{\ N}\ =\ \partial_M\xi^N.
		\ee
		Generalized Lie derivatives define a Lie algebra under commutation. We have
		\begin{eqnarray}
			\left[\gld{\xi_1}\ ,\ \gld{\xi_2}\right]=\  \gld{\lb[\xi_1,\ \xi_2\rb]_C},\label{gldLA}
		\end{eqnarray}
		where $\lb[\cdot , \cdot \rb]_C$ is the C-Bracket defined as follows.
		\be 
		\lb[\xi_1,\xi_2\rb]_{C}^{M}=\xi_1^P\ \partial_M\ \xi_2^P-\xi_2^P\ \partial_P\ \xi_1^M-\frac{1}{2}\lb(\xi_1^P\ \partial^M \xi_{2P}-\xi_2^P\ \partial^M\ \xi_{1P}\rb). \label{Cbrack}
		\ee
		The Jacobi identity also holds for the commutation relation (\ref{gldLA}) because the Jacobiator of three parameters $\xi_1,\ \xi_2,\ \xi_3$ is a trivial parameter and the generalized Lie derivative of a trivial parameter vanishes \cite{GenMetricForm,GaugeAlgebra}. This allows us to realize finite transformations, which form a group, by exponentiation of the generalized Lie derivative
		\be 
		V^{\prime}(X)=e^{\hat{\mathcal{L}}_{\xi}}\ V(X).\label{tranI}
		\ee
		By using the fact  $\lb(e^{\xi}f\rb)(X^{\prime})=f(X)$ \cite{LargeGT}, we can write
		\begin{eqnarray}
			V^{\prime}(X)&=&e^{\xi} V^{\prime}(X^{\prime})=e^{\xi} \mathcal{F} V(X) .\label{tranII}
		\end{eqnarray}
		If we na\"ively compare (\ref{tranI}) and (\ref{tranII}), we get the equality in (\ref{conj1}). However, it can be seen easily that the equality holds only for a very special class of parameters: the quasi-trivial parameters. For an arbitrary parameter, it holds only up to the second order in the parameter. However, this does not lead to any inconsistency as argued in \cite{LargeGT}. Due to the strong constraint, double field theory allows some extra freedom. This freedom can be exploited to change the parameter so that the transformation works out for the vector field. Transformation properties of the scalar field put some restrictions on the nature of the modified parameter. By definition, under any coordinate transformation $X\ \mapsto\  X'$,  a scalar field $\phi(X)$ transforms as
		\begin{eqnarray}
			\phi'\lb(X'\rb)&=&\phi\lb(X\rb),
		\end{eqnarray}
		which implies
		\begin{eqnarray}
			\phi'\lb(X\rb)\ =\ e^{\xi}\ \phi\lb(X\rb).\label{ScalarFinite}
		\end{eqnarray}
		We now see how the finite transformation of the scalar field is realized by exponentiating the generalized Lie derivative. Then we compare it with the transformation rule (\ref{ScalarFinite}). The action of the generalized Lie derivative on a scalar field $\phi(X)$ is given by 
		\begin{eqnarray}
			\gld{\xi}\phi(X)\ =\ \xi^P\partial_P\phi(X).
		\end{eqnarray}
		Therefore,
		\begin{eqnarray}
			\phi'(X)\ =\ e^{\hat{\mathcal{L}}_{\xi}}\phi(X)\ =\ e^{\xi}\phi(X),
		\end{eqnarray}
		which is consistent with (\ref{ScalarFinite}). We want the modified parameter $\xi'(\xi)$ to preserve this consistency, i.e., 
		\begin{eqnarray}
			e^{\xi'\lb(\xi\rb)}\phi'\lb(X\rb)\ =\ e^{\xi}\phi\lb(X\rb)
		\end{eqnarray}
		This condition can be satisfied if the following holds
		\begin{eqnarray}
			\xi'^P\partial_P \lb(\phi(X)\rb)=\xi^P\partial_P\lb(\phi(X)\rb).
		\end{eqnarray}
		It is possible to satisfy the above requirement if $\xi'\lb(\xi\rb)$ and $\xi$ differ by a quasi-trivial parameter, as defined earlier, and we aim to develop a procedure which determines such $\xi'\lb(\xi\rb)$ to any desired order in $\xi$.
		\subsection{Decomposition of $\mathcal{G}\lb(\xi'\lb(\xi\rb)\rb)$}\label{decG}
			
		In this section we will prove important results which will help us in deriving a differential equation for $\mathcal{G}\lb(\xi'\lb(\lambda\xi\rb)\rb)$ in section \ref{deqG}.  We will demonstrate how $\mathcal{G}\lb(\xi'(\xi)\rb)$ can be written as a product of $\mathcal{G}\lb(\xi\rb)$ and $\mathcal{G}\lb(\chi(\xi)\rb)$, where $\chi\lb(\xi\rb)$ is a quasi-trivial parameter. We will also derive an expression for $\xi'(\xi)$ in terms of $
		\xi$ and $\chi\lb(\xi\rb)$.
		\par For a quasi-trivial parameter $\delta$ we can use the strong constraint to obtain
		\be 
		\mathcal{G}(\delta)=e^{-\delta}e^{\gld{\delta}} \ =\ e^{\gld{\delta}} . \label{Gqt}
		\ee
		Note that due to the strong constraint we have $\gld{\delta}\ =\ \delta^P\partial_P +b-b^t\ =\ b-b^t\ $, where $b_{M}^{\ N}\ \equiv\ \partial_{M} \delta^{N}$. Both indices of $b$ are carried by derivatives, therefore, $b$, $b^t$ and $b-b^t$ are all nilpotent matrices and we can simplify equation (\ref{Gqt}) to get
		\be 
		\mathcal{G}\lb(\delta\rb)\ =\ e^{b-b^t}\ =\ \mathbf{1} + \ b-\ b^t\label{Gqt2}.
		\ee
		The relations (\ref{Gqt}) and (\ref{Gqt2}) for a quasi-trivial parameter will be used extensively in the rest of this paper.  Let us consider a collection of quasi-trivial parameters $ \delta_{1},\delta_{2},\cdots,$ and define $\lb(b_{i}\rb)_{M}^{\ N}\ \equiv\ \partial_{M} \delta_{i}^{N}$.  Using the strong constraint,  it is easy to see that the following identities hold. 
		\be  
		b_{i}b_{j}\ =\ b_{i}b_{j}^{t}\ =\ b_{i}^tb_{j}\ =\ b_{i}^tb_{j}^t\ =\ 0, \ \ \forall\ \ i,j\ .\ee
		Now we consider the product $\mathcal{G}\lb(\delta_1\rb)\mathcal{G}\lb(\delta_2\rb)\cdots $. By using equation (\ref{Gqt2}) this product can be written as follows.
		\begin{eqnarray}
			\mathcal{G}\lb(\delta_1\rb)\mathcal{G}\lb(\delta_2\rb)\cdots \ &=& \lb(1 + b_{1}-b_{1}^t\rb)\lb(1 + b_{2}-b_{2}^t\rb)\cdots\ , \\
			&=&\  1+b_{1}-b_{1}^t+b_{2}-b_{2}+\cdots\ , \\
			&=&\ e^{b_{1}+b_{2}+\cdots-\lb(b_{1}^t+b_{2}^t+\cdots\rb)}.
		\end{eqnarray}
		Since all $\delta_i$'s are quasi-trivial, the relation (\ref{Gqt}) implies that
		\be 
		\mathcal{G}\lb(\delta_1\rb)\mathcal{G}\lb(\delta_2\rb)\cdots \ =\ \mathcal{G}\lb(\delta_1+\delta_2+\cdots\rb).\label{lem1}
		\ee
		The equation (\ref{lem1}) will be useful in establishing our main result for this section.
		\begin{thm}\label{decThm}
			We have
			$
			\xi'\ =\ \xi+\delta$  with  $\delta$ quasi-trivial
			if and only if 
			\be  
			\mathcal{G}\lb(\xi'\rb)=\mathcal{G}(\xi)\mathcal{G}\lb(\chi\rb),\label{GpAnsatz}
			\ee
			with $\chi$ also quasi-trivial. Moreover,
			\be 
			\delta\ =\ \chi+\frac{1}{2}\lb[\xi,\chi\rb]_{C}+\frac{1}{12}\left(\lb[\xi,\lb[\xi,\chi\rb]_C\rb]_C+\lb[\chi,\lb[\chi,\xi\rb]_C\rb]_C\right)+\cdots\label{NewParDef}
			\ee
			up to the addition of trivial parameters. Here $\lb[\ .\ ,\ .\rb]_C$ is the \textit{C-Bracket} and `$\cdots$' contains nested C-brackets between $\xi$ and $\chi$ given by the BCH formula.\label{DecompThm}
		\end{thm}
		\begin{proof}
					
			The proof consists of two parts. In the first part, we assume that $\xi'$ and $\xi$ differ by a quasi-trivial parameter and show that this leads to the decomposition (\ref{GpAnsatz}).
			Using the definition (\ref{defG}), we can write
			\begin{eqnarray}
				\mathcal{G}\lb(\xi'\rb)\ &=&\ e^{-\xi'}\ e^{\gld{\xi'}}\ =\ e^{-\xi-\delta}\ e^{\gld{\xi+\delta}}\ =\ e^{-\xi}\ e^{\gld{\xi}\ +\ \gld{\delta}},
			\end{eqnarray}
			where the last equality follows due to the strong constraint.
			We make use of the `Zassenhaus' formula \cite{Zhaus}, which can be written as follows.
			\be 
			e^{X+Y}\ =\ e^X\ e^Y \prod_{i\geq2} e^{C_{i}\lb(X,Y,\lb[\cdot,\cdot\rb]\rb)},
			\ee
			where $C_{i}\lb(X,Y,\lb[\cdot,\cdot\rb]\rb)$ is the $i$\textsuperscript{th} order Lie polynomial in $X$ and $Y$. The precise form of $C_{i}\lb(X,Y,\lb[\cdot,\cdot\rb]\rb)$ is not essential for our arguments here. However, the fact that it only involves commutators between $X$ and $Y$ is crucial. $C_{i}\lb(X,Y,\lb[\cdot,\cdot\rb]\rb)$ can be determined as explained in \cite{Zhaus}. For example, $C_{2}\lb(X,Y,\lb[\cdot,\cdot\rb]\rb)\ =\ -\frac{1}{2}\lb[X,Y\rb]$ and a higher order $C_{i}\lb(X,Y,\lb[\cdot,\cdot\rb]\rb)$ involves $i$ nested commutators between $X$ and $Y$. So we have
			\begin{eqnarray}
				\mathcal{G}\lb(\xi'\rb)\ =\ e^{-\xi}\ e^{\gld{\xi}}\ e^{\gld{\delta}}\prod_{i\geq2} e^{C_{i}\lb(\gld{\xi},\gld{\delta},\lb[\cdot,\cdot\rb]\rb)}.\label{DecInt}
			\end{eqnarray}
			From the Lie algebra of the generalized Lie derivatives, we know that \begin{eqnarray}
			\lb[\gld{\xi},\gld{\delta}\rb]\ &=&\ \gld{\lb[\xi,\delta\rb]_c}.\end{eqnarray}
			Therefore,
			\begin{eqnarray}
				C_{i}\lb(\gld{\xi},\gld{\delta},\lb[\cdot,\cdot\rb]\rb) \ &=&\ \gld{C_{i}\lb(\xi,\delta,\lb[\cdot,\cdot\rb]_c\rb)},
			\end{eqnarray} so we can write $\mathcal{G}\lb(\xi'\rb)$ as follows.
			\be 
			\mathcal{G}\lb(\xi'\rb)\ =\ e^{-\xi}\ e^{\gld{\xi}}\ e^{\gld{\delta}}\ \prod_{i\geq2} e^{\gld{C_{i}\lb(\xi,\delta,\lb[\cdot,\cdot\rb]_c\rb)}}.\label{DecInter}
			\ee
			It was shown in section 5.1 of \cite{LargeGT} that the C-bracket involving a quasi-trivial parameter is also quasi-trivial. This means that $\delta$ as well as all of the $C_{i}\lb(\xi,\delta,\lb[\cdot,\cdot\rb]_c\rb)$ are quasi-trivial and we can use the equation (\ref{Gqt}) to write
			\be 
			\mathcal{G}\lb(\xi'\rb)\ =\ \mathcal{G}\lb(\xi\rb)\ \mathcal{G}\lb(\delta\rb)\ \prod_{i\geq2} \mathcal{G}\lb(C_{i}\lb(\xi,\delta,\lb[\cdot,\cdot\rb]_c\rb)\rb).
			\ee
			Now, the result in equation (\ref{lem1}) implies that
			\begin{eqnarray}
				\mathcal{G}\lb(\xi'\rb)\ &=&\ \mathcal{G}\lb(\xi\rb)\ \mathcal{G}\big(\delta+ \sum_{i\geq2} C_{i}\lb(\xi,\delta,\lb[\cdot,\cdot\rb]_c\rb)\big), \\
				&\equiv&  \mathcal{G}\lb(\xi\rb)\ \mathcal{G}\lb(\chi \rb).
			\end{eqnarray}
			In the last step we have defined $\chi\ \equiv\ \delta\ + \sum_{i\geq2} C_{i}\lb(\xi,\delta,\lb[\cdot,\cdot\rb]_c\rb)$, which is clearly quasi-trivial. This completes the first part of the proof.
					
			Now we turn to the second part of the proof. We assume that the composition in (\ref{GpAnsatz}) holds and show that this leads to an expression for $\xi'$, which differs from $\xi$ by a quasi-trivial parameter, i.e., we show that $\delta$ is quasi-trivial.
			From the definition (\ref{defG}), the product $\mathcal{G}(\xi)\mathcal{G}(\chi)$ can be written as 
			\be 
			\mathcal{G}(\xi)\mathcal{G}(\chi)\ =\ e^{-\xi}e^{\gld{\xi}}e^{-\chi}e^{\gld{\chi}}\ =\ e^{-\xi}e^{\gld{\xi}}e^{\gld{\chi}},\label{Gform2}
			\ee
			where the second equality follows from the strong constraint. Now we use the well known Baker-Campbell-Hausdorff formula to write
			\be \label{HZId}
			e^{\gld{\xi}}e^{\gld{\chi}}\ =\ e^{\gld{\xi'}},
			\ee
			where 
			\be 
			\gld{\xi'}=\gld{\xi}+\gld{\chi}+\frac{1}{2}\left[\gld{\xi},\gld{\chi}\right] + \frac{1}{12}\left(\lb[\gld{\xi},\lb[\gld{\xi},\gld{\chi}\rb]\rb]+\lb[\gld{\chi},\lb[\gld{\chi},\gld{\xi}\rb]\rb]\right)\cdots,
			\ee
			where 	`$\cdots$' contains further nested commutators between $\gld{\xi}$ and $\gld{\chi}$. From the defining relation of the Lie algebra of the generalized Lie derivatives (\ref{gldLA}), we deduce that
			\be 
			\xi'\ =\ \xi+\chi+\frac{1}{2}\left[\xi,\chi\right]_c + \frac{1}{12}\left(\lb[\xi,\lb[\xi,\chi\rb]_c\rb]_c+\lb[\chi,\lb[\chi,\xi\rb]_c\rb]_c\right)+\cdots\ ,\label{detxip}
			\ee
			up to addition of trivial parameters. Here, `$\cdots$' contains further nested C-brackets between $\xi$ and $\chi$. We also note that in the above expression, $\xi'$ and $\xi$ differ by a quasi-trivial parameter because $\chi$ and C-brackets involving $\chi$ are all quasi-trivial. Using this result, we can write
			\begin{eqnarray}
				\mathcal{G}(\xi)\mathcal{G}(\chi)\ &=&\ e^{-\xi}e^{\gld{\xi'}}.
			\end{eqnarray}
			Due to the strong constraint, $ \xi'^P\partial_P=\xi^P\partial_P,$  and we have 
			\begin{eqnarray}
				\mathcal{G}(\xi)\mathcal{G}(\chi)\ &=&\ e^{-\xi'}e^{\gld{\xi'}}\ =\ \mathcal{G}\lb(\xi'\rb),
			\end{eqnarray}
			with $\xi' \ =\ \xi\ +\ \delta$. $\delta$ can be read off from (\ref{detxip}) and matches with the expression in equation (\ref{NewParDef}) up to addition of trivial parameters.
			This completes the proof of theorem (\ref{decG}).
					
		\end{proof}
			
		We note that  in the context of the conjecture (\ref{claim}), both $\delta$ and $\chi$ should be understood as functions of $\xi$, however results of this section hold in general. We conclude this section by simplifying $\mathcal{G}\lb(\xi'\rb)$ further. Since $\chi$ is a quasi-trivial parameter, we can use relation (\ref{Gqt2}) to write
		\begin{eqnarray}
			\mathcal{G}\lb(\chi\rb)&=& 1+\Delta\ -\ \Delta^t \ \text{ with }\ \Delta_{M}^{\;\;N}\ \equiv\partial_M\chi^N,\label{GchiDel}
		\end{eqnarray}
		so that
		\be 
		\mathcal{G}\lb(\xi'\rb)\ =\ \mathcal{G}\lb(\xi\rb)\ +\ \mathcal{G}\lb(\xi\rb)\lb(\Delta -\Delta^t \rb).\label{GxDel}
		\ee
			
		We mentioned in the introduction that the conjecture (\ref{claim}) will be proven by finding the parameter $\xi'\lb(\xi\rb)$. Equation (\ref{GxDel}) implies that our problem of finding $\xi'(\xi)$ has now reduced to finding $\Delta(\xi)$ or $\chi(\xi)$. If we can find a quasi-trivial $\chi\lb(\xi\rb)$, then the theorem proven in this section would guarantee that the HZ conjecture (\ref{claim}) holds.
		\section{Obtaining differential equations}\label{des}
		\subsection{Differential equation for $\mathcal{F}(\lambda\xi)$}\label{difF}
			
		In this section we will derive a differential equation for $\mathcal{F}$. To do this, we introduce a $\lambda$ dependence in the coordinate transformations by letting $\xi \rightarrow \lambda\xi$ so that we have
		\begin{eqnarray}
			X^{\prime^M}&=&X^M-\zeta^M(\lambda\xi)\  \equiv\  e^{-\lambda\xi^P\partial_P}X^M\label{XprimeL},
		\end{eqnarray}
		where the last equality defines $\zeta^M$ as a function of $\lambda$. Consequently, the matrices $A$ and $B$ defined earlier also become functions of $\lambda$ and we write
		\begin{eqnarray}
			A_M^{\;\;N}(\lambda)&\equiv&\partial_M\zeta^N\lb(\lambda\xi\rb),
			\label{Alambda}
		\end{eqnarray}
		and 
		\begin{eqnarray}
			B(\lambda)&\equiv&\frac{\partial X}{\partial X^{\prime}}=\frac{1}{1-A(\lambda)}\ =\ \lb(\mathbf{1} \cdot e^{\ \lambda\ \lb(-\overleftarrow{\xi}+a\rb)}\rb).
		\end{eqnarray}
		Here and in the rest of the paper, the matrix $a$ is to be understood as $a_{M}^{\ N}\ =\ \partial_M\xi^N$, i.e., without the $\lambda$ dependence because it has already been taken into account. The matrix $\mathcal{F}$ also becomes $\lambda$-dependent and we can express it as follows:
		\begin{eqnarray}
			\mathcal{F}(\lambda\xi)&=&\frac{1}{2}\left(B(\lambda)\lb(B^t(\lambda)\rb)^{-1}+\lb(B(\lambda)^t\rb)^{-1}B(\lambda)\right).
		\end{eqnarray}
			
		Now we are in a position to derive a first order differential equation for $\mathcal{F}(\lambda\xi)$ in $\lambda$. For brevity, explicit $\lambda$ dependence is suppressed in this section. From the definition of $B$, it is easy to see that $B(\lambda)$
		and $\lb(B^{t}\rb)^{-1}\lb(\lambda\rb)$ satisfy the differential equations
		\begin{eqnarray}
			\frac{dB}{d\lambda}\ =\ B\left(-\overleftarrow{\xi}+a\right)\ \ \ \text{and}\ \ \ 
			\frac{d\lb(B^{t}\rb)^{-1}}{d\lambda}\ &=&\ \lb(B^{t}\rb)^{-1}\left(-\overleftarrow{\xi}-a^t\right).
		\end{eqnarray}
		Straightforward applications of the chain rule then yields
		\begin{eqnarray}
			\frac{d}{d\lambda}\left(B\lb(B^{t}\rb)^{-1}\right)
			&=B\lb(B^{t}\rb)^{-1}\lb(-\overleftarrow{\xi}\rb)+B\ a\lb(B^{t}\rb)^{-1}+ B\lb(B^{t}\rb)^{-1}(a-a^t)-B\lb(B^{t}\rb)^{-1}a, 
			\\
			\frac{d}{d\lambda}\left(\lb(B^{t}\rb)^{-1}B\right)&=\lb(B^{t}\rb)^{-1} B\lb(-\overleftarrow{\xi}\rb)+\lb(B^{t}\rb)^{-1} B\lb(a-a^t\rb)+\lb(B^{t}\rb)^{-1}Ba^t-\lb(B^{t}\rb)^{-1}a^tB.
		\end{eqnarray}
		Now we make use of the following identities which are a consequence of the strong constraint: \be  a^tB=a^t,\  \lb(B^{t}\rb)^{-1}B=B+\lb(B^{t}\rb)^{-1}-1,\  \lb(B^{t}\rb)^{-1}a=a\cdot\ee Using these, we can simplify the above differential equations and combine them to get a differential equation for $\mathcal{F}(\lambda)$:
		\begin{eqnarray}
			\frac{d}{d\lambda}\mathcal{F}&=&\frac{d}{d\lambda}\frac{1}{2}\left(B\lb(B^{t}\rb)^{-1}+\lb(B^{t}\rb)^{-1}B\right),\\
			&=&\mathcal{F} \lb(-\overleftarrow{\xi}+a-a^t\rb)+\frac{1}{2}\left(B\ a\ \lb(B^{t}\rb)^{-1}-Ba+Ba^t-a^t\right).\label{dFdt1}
		\end{eqnarray}
		We  can also write
		\be 
		\lb(B^t\rb)^{-1} =\lb(\lb(\frac{1}{1-A}\rb)^t\rb)^{-1}=1-A^t,
		\ee
		such that
		\begin{eqnarray}
			Ba\lb(B^{t}\rb)^{-1}\  &=&\ Ba(1-A^t)=Ba-BaA^t,  \\
			Ba^t-a^t&=&B(a^t-B^{-1}a^t)=B(a^t-(1-A)a^t)=BAa^t. 
		\end{eqnarray}
		Plugging these relations into (\ref{dFdt1}), we finally arrive at a differential equation for $\mathcal{F}(\lambda)$.
		\begin{eqnarray}
			\boxed{\frac{d}{d\lambda}\mathcal{F}=\mathcal{F}(-\overleftarrow{\xi}+a-a^t)+\frac{1}{2}\ B\ \left(-aA^t+Aa^t\right).\label{dFdt2}}
		\end{eqnarray}
		This equation will play an important role in our analysis in rest of the paper. We will later obtain a differential equation for $\mathcal{G}$ which has a very similar form. We also re-write this differential equation in a more visually appealing form by using the strong constraint. By the strong constraint, $B^{t^{-1}}A=A$ and $\lb(B^{t}\rb)^{-1}a=a$ . This allows us to write
		\begin{eqnarray}
			\mathcal{F}\left(Aa^t-aA^t\right)&=&\frac{1}{2}\left(B\lb(B^{t}\rb)^{-1}(Aa^t-aA^t)+\lb(B^{t}\rb)^{-1}B(Aa^t-aA^t)\right), \\
			&=&\frac{1}{2}\left(B(Aa^t-aA^t)+(B+\lb(B^{t}\rb)^{-1}-1)(Aa^t-aA^t)\right), \\
			&=&\frac{1}{2}\left(B(Aa^t-aA^t)+Aa^t-Aa^t+B(Aa^t-aA^t)-Aa^t+aA^t\right), \\
			&=&B\left(Aa^t-aA^t\right).
		\end{eqnarray}
		Using this result in (\ref{dFdt2}) we get
		\begin{equation}
			\begin{split}
				\frac{d}{d\lambda}\mathcal{F}\ =\  &\mathcal{F}\ \left(-\overleftarrow{\xi}+a-a^t+\frac{1}{2}Aa^t-\frac{1}{2}aA^t\right).\label{dFdt3}
			\end{split}
		\end{equation}
		\par There appears to be no systematic way of solving the differential equation (\ref{dFdt3}) to obtain a closed form expression for $\mathcal{F}\lb(\lambda\xi\rb)$.
		\subsection{Differential Equation for $\mathcal{G}\lb(\xi'\lb(\lambda\xi\rb)\rb)$}\label{deqG}
		In this section, we derive a first order differential equation for $\mathcal{G}(\xi'(\lambda))$ by a procedure which is similar to the one used in section (\ref{difF}). We introduce a parameter dependence by $\xi\rightarrow \lambda \xi$.  Doing this makes $\xi^{\prime}(\xi)$, $\chi(\xi)$, $\Delta(\xi)$ and hence $\mathcal{G}\lb(\xi'\lb(\xi\rb)\rb)$ $\lambda$-dependent. We use the theorem (\ref{decThm}) and write the relations (\ref{GpAnsatz}) and (\ref{GchiDel}) with a $\lambda$-dependence.
		\begin{eqnarray}
			\widetilde{G}\lb(\lambda\rb)\ &\equiv &\  \mathcal{G}\lb(\xi'\lb(\lambda\xi\rb)\rb)\ =\ \mathcal{G}(\lambda\xi)\ \mathcal{G}\lb(\chi\lb(\lambda\xi\rb)\ \rb),\\
			\mathcal{G}\lb(\chi\lb(\lambda\xi\rb)\rb)\ &=&\  1+\Delta(\lambda\xi)-\Delta^t(\lambda\xi),\label{GchiDelL} \label{GpAnsatzL}
		\end{eqnarray}
		where we also made a slight change of notation by defining $\widetilde{G}(\lambda)$ .  Now we are in a position to obtain a differential equation for $\widetilde{G}(\lambda)$. We will suppress the explicit  $\lambda$ dependence from here on. Using the strong constraint, it is easy to show that the following identities hold.
		\be 
		a^t\Delta=a^t\Delta^t=\Delta a= \Delta^t a=0,\ \ 
		\mathcal{G} \Delta=B\Delta,\ \ 
		\mathcal{G}\Delta^t=B\Delta^t.
		\ee
		Since these identities are a consequence of the strong constraint and the index structure of the matrices $\Delta$ and $\Delta^t$, they also hold for their derivatives, $\frac{d\Delta}{d\lambda}$ and $\frac{d\Delta^t}{d\lambda}$ .
		Now we differentiate $\widetilde{\mathcal{G}}\lb(\lambda\rb)$ with respect to $\lambda$ to get
		\begin{eqnarray}
			\frac{d\widetilde{\mathcal{G}}}{d\lambda}&=&\frac{d\mathcal{G}}{d\lambda}\left(1+\Delta-\Delta^t\right)+\mathcal{G}\frac{d}{d\lambda}\left(\Delta-\Delta^t\right).
		\end{eqnarray}
		By letting $\xi\rightarrow \lambda\xi$ in equation(\ref{diffformG}) we see that,
		\be 
		\frac{d}{d\lambda}\mathcal{G}=\mathcal{G}\lb(-\overleftarrow{\xi}+a-a^t\rb).
		\ee
		Using this and the strong constraint identities recorded above, we obtain
		\begin{eqnarray}
			\frac{d\widetilde{\mathcal{G}}}{d\lambda}\ =\  \mathcal{G}\left(-\overleftarrow{\xi}+a-a^t \right)\left(1+\Delta-\Delta^t\right)+B\frac{d}{d\lambda}\left(\Delta-\Delta^t\right).
			\label{dgdleq1}
		\end{eqnarray}
		We want to put equation(\ref{dgdleq1}) in a form similar to (\ref{dFdt2}). To do that we need to compute $ \widetilde{\mathcal{G}} \left(-\overleftarrow{\xi}+a-a^t\right)$. After a straightforward but somewhat tedious computation one finds that
		\begin{eqnarray}
			\widetilde{\mathcal{G}} \left(-\overleftarrow{\xi}+a-a^t\right)
			=&&\mathcal{G}\left(-\overleftarrow{\xi}+a-a^t\right)\left(1+\Delta-\Delta^t\right)-B\left(\xi+a\right)\left(\Delta-\Delta^t\right)\nonumber \\ &&-B\left(\Delta-\Delta^t\right)a^t.
		\end{eqnarray}
		Using the last result in equation (\ref{dgdleq1}), we get the desired differential equation
		\begin{equation}
			\boxed{\frac{d\widetilde{\mathcal{G}}}{d\lambda}=\widetilde{\mathcal{G}}\left(-\overleftarrow{\xi}+a-a^t\right)+B\left(\left(\xi+a+\frac{d}{d\lambda}\right)\left(\Delta-\Delta^t\right)+ \left(\Delta-\Delta^t\right)a^t\right). \label{dgdl2}}
		\end{equation}
		Note that this equation has very similar form as the differential equation for $\mathcal{F}\lb(\lambda\rb)$ (\ref{dFdt2}). Now, our goal is to find $\Delta$ such that the two equations are exactly the same. This is the subject of the next section.
		\section{Solving for the parameter $\xi'\lb(\xi\rb)$}\label{solveXip}
		\subsection{Comparison and an iterative method}\label{Method}
			
		It is easy to see that $\mathcal{F}(\lambda\xi)$ and $\mathcal{G}\lb(\xi'\lb(\lambda\xi\rb)\rb)$ satisfy the same initial condition at $\lambda=0$, i.e., $\mathcal{F}(0)\ =\ \mathcal{G}\lb(\xi'\lb(0\rb)\rb)\ =\ \mathbf{1}$. For $\mathcal{F}(\lambda\xi)$ to be equal to $\mathcal{G}\lb(\xi'\lb(\lambda\xi\rb)\rb)$, they must satisfy the same differential equation. We compare the two differential equations, (\ref{dFdt2}) and (\ref{dgdl2}) to get
		\begin{eqnarray}
			\left(\xi+a+\frac{d}{d\lambda}\right)\left(\Delta-\Delta^t\right)+ \left(\Delta-\Delta^t\right)a^t=\frac{1}{2}\left(Aa^t-aA^t\right).\label{condDelta}
		\end{eqnarray}
		Without loss of generality we take $\Delta(\lambda)$ to be of the following form:
		\begin{eqnarray}
			\Delta(\lambda)=\SUL{n=1}{\infty}\lambda^n \Delta_{n},\label{deltaForm}
			\end {eqnarray}
			where $\Delta_{n}$ is a $\lambda$-independent term which is $n$\textsuperscript{th} order in $\xi$. Also it is obvious that we can't have a zeroth order term in $\Delta$.  Following the same notation we write
			\begin{eqnarray}
				A(\lambda)& \equiv &\SUL{n=1}{\infty}\lambda^n A_{n},
			\end{eqnarray}
			where $A_n$ can be read off from (\ref{XprimeL}) and (\ref{Alambda}). It takes the following form.
			\begin{equation}
				\lb(A_n\rb)_{M}^{\ \ N}=\partial_{M}\left(\zeta_{(n)}^N\right).\label{An}
			\end{equation}
			Here $\zeta_{(n)}^N$ is the $n$\textsuperscript{th} order term in $\xi$. We can read off $\zeta_{\lb(n\rb)}^N$ from equation (\ref{defA}) and \be \zeta_{\lb(n\rb)}^N\ =\ \frac{(-1)^{n+1}}{n!}\lb(\left(\xi^P\ \partial_P\right)^{n-1}\ \xi^N\rb).\ee
			Now, from (\ref{condDelta}) we get
			\begin{eqnarray}
				\boxed{(n+1)\left(\Delta_{n+1}-\Delta_{n+1}^{t}\right)+\left(\xi+a\right)\left(\Delta_n-\Delta_{n}^{t}\right)+\left(\Delta_{n}-\Delta_{n}^{t}\right)a^t=\frac{1}{2}\left(A_na^t-aA_{n}^{t}\right).}\label{Iterative}
			\end{eqnarray}
			The last expression is an iterative equation for $\Delta_{n}-\Delta_{n}^t$. There is a systematic procedure which can be used to obtain  $\xi^{\prime}$ to all orders. Before describing our procedure, we want to show that in general, $\Delta_n-\Delta_{n}^t$ obtained from equation (\ref{Iterative}) will be of the form
			\be 
			\lb(\Delta_n-\Delta_n^t\rb)_{M}^{\ \ N}=
			\partial_M f \partial^Ng-\partial_M g \partial^Nf,\label{delform}
			\ee
			where $f$ and $g$ are $\xi$-dependent quantities and they are of the form, $f\sim\xi^k\cdot\ \xi^P$\ and $g\sim \xi^{n-k-2}\cdot\xi_P$, where the powers are chosen so that the whole term is n\textsuperscript{th} order in $\xi$. 
					
			To show this, first note that from equation (\ref{Iterative}), it follows trivially that $\Delta_1-\Delta_{1}^t=\Delta_2-\Delta_{2}^t=0$, which can certainly be written in the form (\ref{delform}). Now we use inductive argument and suppose that $\Delta_{n}-\Delta_{n}^t$ is of the form (\ref{delform}), i.e.,
			\be 
			\Delta_n-\Delta_{n}^t\ \sim\ \partial_M\lb(\xi^k\cdot\ \xi^P\rb)\partial^N\lb( \xi^{n-k-2}\cdot\ \xi_P\rb) -\partial_M\lb(\xi^{n-k-2}\cdot\ \xi^P\rb)\partial^N\lb(\xi^k\cdot\ \xi^P\rb), 
			\ee
			then a short calculation yields
			\begin{eqnarray}
				&&\lb(\xi+a\rb)\lb(\Delta_n-\Delta_{n}^t\rb)+\lb(\Delta_n-\Delta_{n}^t\rb)a^t \\ &\sim&\partial_M\lb(\xi^{k+1}\cdot\xi^P\rb)\partial^N\lb(\xi^{n-k-2}\cdot\xi_P\rb) -\partial_M\lb(\xi^{n-k-2}\cdot\xi^P\rb)\partial^N\lb(\xi^{k+1}\cdot\xi^P\rb)  \\ &+& \partial_M\lb(\xi^k\cdot\xi^P\rb)\partial^N\lb( \xi^{n-k-1}\cdot\xi_P\rb) -\partial_M\lb(\xi^{n-k-1}\cdot\xi^P\rb)\partial^N\lb(\xi^k\cdot\xi^P\rb), \\
				&\sim&\  \partial_M\lb(\xi^k\cdot\xi^P\rb)\partial^N\lb( \xi^{n-k-1}\cdot\xi_P\rb) -\partial_M\lb(\xi^{n-k-1}\cdot\xi^P\rb)\partial^N\lb(\xi^k\cdot\xi^P\rb),
			\end{eqnarray}
			which is of the form (\ref{delform}). We can also read off the RHS of equation (\ref{Iterative}) and see that:
			\be 
			A_{n}a^t- a A_{n}^t\ \sim\ \partial_M\lb(\xi^{n-1}\cdot\xi^P\rb)\partial^N\xi_P-\partial_M\xi^P\partial^N\lb(\xi^{n-1}\cdot\xi_P\rb).
			\ee
			This is also of the form (\ref{delform}). Since all the terms in (\ref{Iterative}) are of the form (\ref{delform}), we deduce that $\Delta_{n+1}-\Delta_{n+1}^t$ will also be of the form (\ref{delform}) and this completes our inductive argument.
					
			\par In the following, we outline our procedure to obtain $\xi'\lb(\xi\rb)$.
			\begin{itemize}
				\item Using the fact that $\Delta_{0}-\Delta_{0}^t=0$, the equation (\ref{Iterative}) can be solved iteratively to find all $\Delta_{n}-\Delta_{n}^t$, which will be of the form (\ref{delform}) .
				\item Only $\Delta_n-\Delta^{t}_n$ can be found uniquely from equation ( \ref{Iterative}) and there are many different choices for $\Delta_n$ which differ from one another by symmetric matrices. We can choose \be \lb(\Delta_n\rb)_{M}^{\ \ N}=\partial_M\lb(f\partial^Ng\rb).\ee
				\item Using (\ref{deltaForm}) one can obtain $\Delta(\lambda\xi)$ to all orders in $\xi$ and by setting $\lambda=1$, we get $\Delta(\xi)$.
				\item Since all $\Delta_n$'s are total derivatives, $\Delta(\xi)$ will also be a total derivative therefore \be  \Delta_{M}^{\;\;N}=\partial_M\chi^N,\ee can be solved to obtain an expression for $\chi(\xi)$. Note that $\chi\lb(\xi\rb)$ is determined only up to the addition of trivial terms.
				\item By using (\ref{NewParDef}), one can find $\delta(\xi)$ to all orders in $\xi$. Since $\xi'\lb(\xi\rb)\ =\ \xi+\delta\lb(\xi\rb)$, we can finally obtain $\xi'\lb(\xi\rb)$ to all orders in $\xi$ and this was our aim.
			\end{itemize}
			We also note that this procedure does not give a unique $\xi'$ and it is determined only up to the addition of trivial parameters. This non uniqueness arises precisely because of the fact that the generalized Lie derivative with respect to a trivial parameter is zero.
			\subsection{Order by Order checks}\label{checks}
			In this section, we demonstrate how the procedure outlined above can be used systematically to obtain  $\xi^{\prime}\lb(\xi\rb)$ to quintic order in $\xi$. Recall that $\xi'\lb(\xi\rb)\ =\ \xi+\delta\lb(\xi\rb)$, and we only need to find out $\delta\lb(\xi\rb)$. We will show that our results here are in agreement with the computation of $\delta(\xi)$ done previously in \cite{LargeGT} to quartic order. The quintic order result is also verified explicitly.
			\par Before proceeding further, we would introduce some useful notations here.
			\begin{itemize}
				\item We will represent matrices of the form  `$f_M\ g^N$' in an index free manner as follows.
				      \be \lb(B\rb)_{M}^{\ \ N}=f_{M}g^N \ \equiv\ \vec{f}\ \vec{g}, \ee
				      i.e., in an expression, the terms carrying free matrix indices are represented with an overhead arrow. 
				\item If the free vector index of a parameter is carried by a partial derivative, i.e., $V^M= A\  \partial^M \ B,$ we will denote it as follows.
				      \be  V=A\ \vec{\partial}\ B.\ee
				\item We will often use the following short hand notation. \be  A-A^t \equiv \left( A- \left(\cdots \right)^t\right).\ee
				\item We will use $\chi_n$ to denote the $\bigO(\xi^n)$ term appearing in $\chi(\xi)$, i.e., 
				      \be 
				      \chi(\xi)\equiv \chi_1+\chi_2+\chi_3 + \cdots\ .
				      \ee
				      It is then trivial to see that $\lb(\Delta_{n}\rb)_{M}^{\ N}=\partial_M \chi_{n}^{\ N}$.
				\item We will compute $\delta(\xi)$ upto the quintic order. In the analysis to follow, we will also find that $\chi_1=\chi_2=0$. Using this and the equation (\ref{NewParDef}) we can write $\xi'(\xi)$ up to quintic order as follows.
				      \be  \label{xiPQuint}
				      \delta(\xi)=\chi_3+\left(\chi_4 +\frac{1}{2}\left[\xi,\chi_3\right]_c\right)+\left(\chi_5 +\frac{1}{2}\left[\xi,\chi_4\right]_c+\frac{1}{12}\left[\xi,\left[\xi,\chi_3\right]_c\right]_c\right) + \bigO(\xi^6)\ .
				      \ee
			\end{itemize}
			\subsubsection{First and Second Order}
			First and second order checks are trivial. By expanding $\mathcal{G}(\xi)$ and $\mathcal{F}(\xi)$ to second order it can bee seen that they are equal. We will obtain the same conclusion by using our procedure.
			\par Using the fact that $A$ does not have a zeroth order term, i.e., $A_{0}$ vanishes, it is trivial to see that $\Delta_{1}-\Delta_{1}^t=0$ from (\ref{Iterative}). Without loss of generality we can choose $\Delta_1\ =\ 0\ \Rightarrow \  \chi_1\ =\ 0$, which implies that
			\begin{eqnarray}
				\delta^{ M}\lb(\xi\rb)&=& 0+ \bigO(\xi^2)\ .
			\end{eqnarray}
			For the second order case, we see that for $n=1$, equation (\ref{Iterative}) becomes 
			\be  
			2\left(\Delta_2-\Delta_{2}^{t}\right)=\frac{1}{2}\left(A_1a^t-aA_{1}^t\right)\ .
			\ee
			From (\ref{An}) we see that
			$ A_1=a
			$, which implies that
			$
			\Delta_2-\Delta_{2}^{t}=0\ .$ Without loss of generality, we have $\Delta_2=0\Rightarrow \chi_2=0,$ which implies
			\begin{eqnarray}
				\delta^M\lb(\xi\rb)&=&0 + \bigO(\xi^3)\ .
			\end{eqnarray}
			This shows that $\mathcal{F}(\xi)$ and $\mathcal{G}(\xi)$ are equal up to quadratic order in $\xi$, in agreement with \cite{LargeGT}.
			\subsubsection{Cubic Order}
			First and quadratic order computations were trivial in the sense that we found $\delta\lb(\xi\rb)=0+\mathcal{O}(\xi^3)$. Things get more interesting at the cubic order as we will get non-trivial results for $\delta\lb(\xi\rb)$ and hence for $\xi'\lb(\xi\rb)$.
			\begin{itemize}
				\item For $n=2$, from (\ref{Iterative}), we get 
				      \be 
				      3\left(\Delta_3-\Delta_{3}^{t}\right)=\frac{1}{2}\left(A_2a^t-aA_{2}^t\right)\label{iter3}.
				      \ee
				      $A_2$ can be read off from equation (\ref{An}). After a short computation we obtain the following:
				      \begin{eqnarray}
				      	\left(A_2a^t-aA_2^t\right)_{M}^{N}&=&-\frac{1}{2}\left(\partial_M\left(\xi\cdot\xi^P\right)\partial^N\xi_P-\partial_M\xi^P\partial^N\left(\xi\cdot\xi_P\right)\right),\\
				      	&=&-\frac{1}{2}\left(\partial_M\left(\left(\xi\cdot\xi^P\right)\partial^N\xi_P\right)-\partial^N\left(\left(\xi\cdot\xi^P\right)\partial_M\xi_P\right)\right),\label{A2at}
				      \end{eqnarray}
				      where in the last step, we added a symmetric part to the both terms so that they can written as total derivatives. Now we use the  notation introduced at the beginning of this subsection to write
				      \begin{eqnarray}
				      	A_2a^t-aA_{2}^t=-\frac{1}{2}\left(\vpa\left((\xi\cdot\xi^P)\vpa\xi_P\right)-\left(\cdots\right)^t\right).
				      \end{eqnarray}
				      Using this in equation (\ref{iter3}) we get:
				      \be 
				      \Delta_3-\Delta_{3}^t=-\frac{1}{12}\left(\vpa\left((\xi\cdot\xi^P)\vpa\xi_P\right)-\left(\cdots\right)^t\right).
				      \ee
				\item Now we can choose
				      \be 
				      \Delta_3=-\frac{1}{12}\vpa\left((\xi\cdot\xi^P)\vpa\xi_P\right)\ \ \Rightarrow\ \  \chi_3=-\frac{1}{12}(\xi\cdot\xi^P)\vpa\xi_P .\label{chi3}
				      \ee
				\item From equation (\ref{xiPQuint}), we can write down $\delta\lb(\xi\rb)$ to cubic order in $\xi$ as follows.
				      \be \delta^{M}(\xi)=-\frac{1}{12}(\xi\cdot\xi^P)\partial^M\xi_P +\bigO(\xi^4) .\ee 
				      So we deduce that
				      \be 
				      \mathcal{F}(\xi)=\mathcal{G}\lb(\xi-\frac{1}{12}(\xi\cdot\xi^P)\vpa\xi_P\rb)+\bigO(\xi^4).\ee
				      This relationship can be inverted easily to get
				      \be 
				      \mathcal{F}\lb(\xi+\frac{1}{12}(\xi\cdot\xi^P)\vpa\xi_P\rb)=\mathcal{G}(\xi)+\bigO(\xi^4).
				      \ee
				      This is in complete agreement with results of section(4) of \cite{LargeGT}.
			\end{itemize}
			\subsubsection{Quartic Order}
			Now we turn to compute $\delta(\xi)$ to quartic order in $\xi$. In \cite{LargeGT}, it was shown by explicit computation  that there is no quartic order term in $\delta(\xi)$. We will see that our procedure reproduces this result through very non trivial cancellation between quartic order terms.
			\begin{itemize}
				\item For $n=3$ the equation (\ref{Iterative}) becomes
				      \be  
				      4\left(\Delta_{4}-\Delta_{4}^{t}\right)+\left(\xi+a\right)\left(\Delta_3-\Delta_{3}^{t}\right)+\left(\Delta_{3}-\Delta_{3}^{t}\right)a^t=\frac{1}{2}\left(A_3a^t-aA_{3}^{t}\right).\label{Iter4}
				      \ee
				      Let us compute different parts of the above equation now. All the computations are fairly straightforward and the final results are presented here.
				      \begin{eqnarray}
				      	A_3a^t-aA_{3}^t&=&\frac{1}{6}\left(\vpa\left((\xi^2\cdot\xi^P)\vpa\xi_P\right)-\left(\cdots\right)^t\right),\label{A3}\\
				      	\left(\xi+a\right) \Delta_3 &=&\vpa\lb(\xi^2\cdot\xi^K\rb)\vpa\xi_K+\vpa\lb(\xi\cdot\xi^K\rb)\vpa\lb(\xi\cdot\xi_K\rb)-\vpa\lb(\xi\cdot\xi^K\rb)\vpa\xi^P\partial_P\xi_K,\\
				      	\lb(\xi+a\rb)\Delta_{3}^t\ &=&\ \vpa\lb(\xi\cdot\xi_K\rb)\vpa\lb(\xi\cdot\xi^K\rb)+\vpa\xi_K\vpa\lb(\xi^2\cdot\xi^K\rb)-\vpa\xi_K\vpa\xi^P\partial_P\lb(\xi\cdot\xi^K\rb),\\
				      	\lb(\Delta_3-\Delta_{3}^{t}\rb)a^t&=& \vpa\lb(\xi\cdot\xi^K\rb)\vpa\xi^P\partial_P\xi_K-\vpa\xi_K\vpa\xi^P\partial_P\lb(\xi\cdot\xi^K\rb),
				      \end{eqnarray}
				      We can determine $\Delta_4-\Delta_{4}^t$ by using these results in equation (\ref{Iter4}).
				      \begin{eqnarray}
				      	\Delta_4-\Delta_{4}^t&=&\frac{1}{24}\lb(\vpa\lb(\lb(\xi^2\cdot\xi^K\rb)\vpa\xi_K\rb)-\lb(\cdots\rb)^t\rb).
				      \end{eqnarray}
				\item Again, without loss of generality, we can choose
				      \be 
				      \Delta_4=\frac{1}{24}\vpa\lb(\lb(\xi^2\cdot\xi^K\rb)\vpa\xi_K\rb)\ \ \ \Rightarrow \chi_4 =\frac{1}{24}\lb(\xi^2\cdot\xi^P\rb)\vpa\xi_P\ . \ee
				      Note that we have a non trivial quartic order term.
				\item Now 
				      \be 
				      \delta(\xi)=\chi_3+\chi_4+\frac{1}{2}\left[\xi,\chi_3\right]_C+\bigO(\xi^5),
				      \ee 
				      the C-bracket, $\left[\xi,\chi_3\right]_C$ can be computed easily by using the value of $\chi_3$ given in equation (\ref{chi3}) and the definition of the C-bracket (\ref{Cbrack}). After a straightforward computation and some index relabeling it is easy to show that
				      \begin{eqnarray}\
				      	\frac{1}{2}\lb[\xi\ , \chi_3\rb]_C&=&-\frac{1}{24}\lb(\xi^2\cdot\xi^P\rb)\vec{\partial}\xi_P =-\chi_4\ .
				      \end{eqnarray}
				      Hence we see that the quartic order contributions from $\chi_4$ and $\frac{1}{2}\lb[\xi,\chi_3\rb]_c$ cancel each other and we get
				      			      
				      \be 
				      \delta^{M}(\xi)=-\frac{1}{12}\lb(\xi\cdot\xi^P\rb)\partial^M\xi_P + \bigO(\xi^5).
				      \ee
			\end{itemize}
			which agrees with the quartic order result of \cite{LargeGT}.
			\subsubsection{Quintic Order}
			Now we turn to the quintic order computation for $\xi'\lb(\xi\rb)$. In contrast with the previous orders, this computation has not been done before. Here we will determine $\xi'(\xi)$ to fifth order in $\xi$ and then check our result by explicitly expanding  $\mathcal{F}(\xi)$ and $\mathcal{G}\lb(\xi'(\xi)\rb)$ to the relevant order. 
					 
			\begin{itemize}
				\item For $n=4$, the equation (\ref{Iterative}) becomes:
				      \be 
				      5\lb(\Delta_5-\Delta_{5}^t\rb)+\lb(\xi+a\rb)\lb(\Delta_4-\Delta_{4}^t\rb)+\lb(\Delta_4-\Delta_{4}^t\rb)a^t=\frac{1}{2}\lb(A_{4}a^t-aA_{4}^{t}\rb).
				      \ee
				      This equation can be solved for $\Delta_5 -\Delta_{5}^{t}$ by similar kind of computation as was done for the quartic order. Key results are summarized below:
				      \begin{eqnarray}
				      	\lb(\xi+a\rb) \ \Delta_4 &=& \frac{1}{24}\lb(\vec{\partial}\ \lb( \xi^3\cdot \xi^P\rb)\vec{\partial}\xi_P+\vec{\partial}\lb(\xi^2\cdot\xi^P\rb)\vec{\partial}\xi_P-\vec{\partial}\lb(\xi^2\cdot\xi^P\rb)\vec{\partial}\ \xi^Q\ \partial_Q\ \xi_P\rb),\ \ \ \ \ \ \ \   \\ 
				      	-\lb(\xi+a\rb) \ \Delta_{4}^t &=&- \frac{1}{24}\lb(\vec{\partial}\ \lb( \xi\cdot\xi^P\rb)\vec{\partial}\lb(\xi^2\cdot\xi_P\rb)+\vec{\partial}\ \xi^P\ \vec{\partial}\lb(\xi^3\cdot\xi_P\rb)\rb)\nonumber \\ && +\frac{1}{24}-\lb(\vec{\partial}\ \xi^P\vec{\partial}\ \xi^Q\ \partial_Q\ \lb(\xi^2\cdot\xi_P\rb)\rb),   \\
				      	\lb(\Delta_4-\Delta_{4}^{t}\rb)\ a^t&=& \frac{1}{24}\lb(\vec{\partial}\lb(\xi^2\cdot\xi^Q\rb)\ \partial^P\ \xi_Q\ \vec{\partial}\ \xi_P-\vec{\partial}\ \xi^Q\ \partial^P\ \lb(\xi^2\cdot \xi_Q\rb)\vec{\partial}\ \xi_P\rb), \\
				      	A_4\ a^t-a\ A_{4}^t &=&-\frac{1}{24}\lb(\vec{\partial}\lb(\xi^3\cdot \xi^P\rb)\vec{\partial}\ \xi_P - \vec{\partial}\  \xi^P\ \vec{\partial}\lb(\xi^3\cdot \xi_P\rb)\rb), \\
				      	\Delta_5-\Delta_{5}^{t}&=&\frac{-1}{240}\lb( 3\ \vec{\partial}\lb(\xi^3\cdot \xi^P\rb)\vec{\partial}\ \xi_P +2\ \vec{\partial}\lb(\xi^2\cdot\xi^P\rb)\vec{\partial}\lb(\xi\cdot\xi_P\rb) -\lb(\cdots\rb)^t\rb)
				      \end{eqnarray}
				\item 
				      Now we follow the familiar procedure and make a choice of $\Delta_5$ which is consistent with the result found for $\Delta_5-\Delta_{5}^t$. We choose \footnote{ A straightforward choice would have been to take \be \Delta_5 =\frac{-1}{240}\vec{\partial}\lb( 3\ \lb(\xi^3\cdot \xi^P\rb)\ \vpa \xi_P +2\ \lb(\xi^2\cdot\xi^P\rb)\vec{\partial}\lb(\xi\cdot\xi_P\rb)\rb).\ee
				      	This choice is perfectly fine but leads to an expression for $\delta(\xi)$ which is rather cumbersome.}
				      \begin{eqnarray}
				      	\Delta_5 = \frac{-1}{240}\vec{\partial}\lb( 3\ \lb(\xi^3\cdot \xi^P\rb)\vec{\partial}\ \xi_P +\ \lb(\xi^2\cdot\xi^P\rb)\vec{\partial}\lb(\xi\cdot\xi_P\rb) - \ \lb(\xi\cdot\xi^P\rb)\vec{\partial}\lb(\xi^2\cdot\xi_P\rb) \ \rb).
				      \end{eqnarray}
				      			      
				\item From $\Delta_5$ we can read off 
				      \begin{eqnarray}
				      	\chi_5&=&\frac{-1}{240}\lb( 3\ \lb(\xi^3\cdot \xi^P\rb)\vec{\partial}\ \xi_P +\ \lb(\xi^2\cdot\xi^P\rb)\vec{\partial}\lb(\xi\cdot\xi_P\rb) - \ \lb(\xi\cdot\xi^P\rb)\vec{\partial}\lb(\xi^2\cdot\xi_P\rb) \ \rb).
				      \end{eqnarray}
				\item We use the fact that $\ \frac{1}{2}\left[\xi,\chi_3\right]_c=-\chi_4$ in equation (\ref{xiPQuint}) to get
				      \begin{eqnarray}
				      	\delta(\xi)=\chi_3+\chi_5+\frac{1}{3}\left[\xi,\chi_4\right]_c+\bigO(\xi^6).
				      \end{eqnarray}
				      After a straightforward computation one finds
				      \begin{eqnarray}
				      	\left[\xi,\chi_4\right]_c  &=&\ \frac{1}{24}\lb(\xi^3\cdot\xi^P\rb)\ \vec{\partial}\xi_P +\frac{1}{48}\lb(\lb(\xi^2\cdot\xi^P\rb)\vec{\partial}\lb(\xi\cdot\xi_P\rb)-\lb(\xi\cdot\xi^P\rb)\vec{\partial}\lb(\xi^2\cdot\xi_P\rb)\rb).
				      \end{eqnarray}
				      Notice that all three terms appearing in $\left[\xi,\chi_4\right]_c$ are also present in $\chi_5$. These terms add up nicely to give
				      \begin{empheq}{align}
				      	\ \ \delta^{M}(\xi)=-\frac{1}{12}\lb(\xi\cdot\xi^P\rb)\partial^M\xi_P +\frac{1}{720}\bigg(& \lb(\xi^3\cdot\xi^P\rb)\ \partial^M\xi_P +\lb(\xi^2\cdot\xi^P\rb)\partial^M\lb(\xi\cdot\xi_P\rb) \nonumber\\&-\lb(\xi\cdot\xi^P\rb)\partial^M\lb(\xi^2\cdot\xi_P\rb)\bigg)+\bigO\lb(\xi^6\rb). 
				      \end{empheq}
			\end{itemize}
			It can be verified explicitly that this value of $\delta(\xi)$  indeed satisfies $\mathcal{F}(\xi)=\mathcal{G}(\xi'(\xi))$ up to quintic order. It was done using symbolic manipulations in Mathematica. Let $\xi_{(n)}^M$ denote the $n$\textsuperscript{th} order term in $\xi'^{M}$. Then the fifth  order contribution in $\mathcal{G}(\xi')$ due to $\xi_{(5)}^M$ is simply given by $\partial_M \xi_{(5)}^N-\partial^N\xi_{(5)M}\equiv \widetilde{\Delta}_{M}^{\ \ N}$, i.e., \be \mathcal{G}(\xi+\xi_{(3)}+\xi_{(5)})=\mathcal{G}(\xi+\xi_{(3)})+\widetilde{\Delta}+\bigO(\xi^6).\ee 
			To show that $\mathcal{F}\lb(\xi\rb)\ =\ \mathcal{G}\lb(\xi+\xi_{(3)}+\xi_{(5)}\rb)+\bigO(\xi^6)\ $, we first expanded $\mathcal{F}\lb(\xi\rb)$ and $\mathcal{G}\lb(\xi+\xi_{(3)}\rb)$ to the quintic order. Then we computed the difference $ \mathcal{F}\lb(\xi\rb)-\mathcal{G}\lb(\xi+\xi_{(3)}\rb)$ up to quintic order and noticed that this difference equals $\widetilde{\Delta}$.
					 
			\section{Determining $\xi^{\prime}$ to all orders} \label{Allord}
					
			From our results in the last section, we have a systematic way of finding $\xi'\lb(\xi\rb)$ to any desired order in $\xi$. This really involves two key steps, first we find $\chi\lb(\xi\rb)$ using the iterative equation (\ref{Iterative}) and then we use the BCH formula (\ref{NewParDef}) to obtain $\xi'\lb(\xi\rb)$. In this section, we complete the first step of this computation to all orders in $\xi$, i.e. we solve for $\chi\lb(\xi\rb)$ to all orders. From equation (\ref{Iterative}), only $\Delta_n-\Delta_{n}^t$ can be obtained uniquely. Here we obtain a formula for $\Delta_n-\Delta_{n}^t$ for any integer $n$. Then $\Delta\lb(\xi\rb)$ and $\chi\lb(\xi\rb)$ can be obtained using their definitions. As argued earlier, $\chi\lb(\xi\rb)$ is determined only up to the addition of trivial parameters. This ambiguity in $\chi\lb(\xi\rb)$ is not important because of the fact that the generalized Lie derivative with respect to a trivial parameter is zero.
			\par From our previous analysis, we expect $\Delta_{n}$ to be an $n$\textsuperscript{th} order matrix function of $\xi$ whose free indices are carried by the partial derivatives. Then it is easy to see the most general form for $\Delta_n-\Delta_{n}^t$ is the following:
			\be 
			\Delta_{n}-\Delta_{n}^t\ =\ \SUL{k=0}{k=n-2} \alpha_{n,k}\lb(\  \vec{\partial}\lb(\xi^{k}\cdot\xi^P\rb)\ \vec{\partial}\lb(\xi^{n-k-2}\cdot \xi_P\rb)\ -\ \lb(\cdots\rb)^t \rb).\label{delAns}
			\ee
			where we have used the notation $A-A^t=\lb(A-\lb(\cdots\rb)^t\rb)$, and $\alpha_{n,k}$ is some $n$ and $k$ dependent coefficient, which we want to find so that the ansatz (\ref{delAns}) satisfies equation (\ref{Iterative}). We record equation (\ref{Iterative}) here:
			\be 
			(n+1)\left(\Delta_{n+1}-\Delta_{n+1}^{t}\right)+\left(\xi+a\right)\left(\Delta_n-\Delta_{n}^{t}\right)+\left(\Delta_{n}-\Delta_{n}^{t}\right)a^t=\frac{1}{2}\left(A_na^t-aA_{n}^{t}\right). \label{Iter2}
			\ee
			Now let us compute different components of the above equation for the ansatz (\ref{delAns}). By straightforward computations one obtains the following.
			\begin{eqnarray}
				\lb(\xi+a\rb)\ \lb(\Delta_{n}-\Delta_{n}^t\rb)\ = \SUL{k=0}{k=n-2} &&\alpha_{n,k}\lb(\  \vec{\partial}\lb(\xi^{k+1}\cdot\xi^P\rb)\ \vec{\partial}\lb(\xi^{n-k-2}\cdot \xi_P\rb)\ -\ \lb(\cdots\rb)^t \rb) \nonumber\\ 
				&+& \alpha_{n,k}\lb(\  \vec{\partial}\lb(\xi^{k+1}\cdot\xi^P\rb)\ \vec{\partial}\lb(\xi^{n-k-1}\cdot \xi_P\rb)\ -\ \lb(\cdots\rb)^t \rb) \nonumber \\
				&-& \alpha_{n,k}\  \vec{\partial}\lb(\xi^{k}\cdot\xi^P\rb)\ \vec{\partial}\xi^Q\partial_Q\lb(\xi^{n-k-2}\cdot \xi_P\rb) \nonumber  \\ &+& \alpha_{n,k}\  \vec{\partial}\lb(\xi^{n-k-2}\cdot\xi^P\rb)\ \vec{\partial}\xi^Q\partial_Q\lb(\xi^{k}\cdot \xi_P\rb),
			\end{eqnarray}
			and 
			\begin{eqnarray}
				\lb(\Delta_n-\Delta_{n}^t\rb) a^t\ =\  \SUL{k=0}{k=n-2}&&\alpha_{n,k}\  \vec{\partial}\lb(\xi^{k}\cdot\xi^P\rb)\ \vec{\partial}\xi^Q\partial_Q\lb(\xi^{n-k-2}\cdot \xi_P\rb) \nonumber \\ &-& \alpha_{n,k}\  \vec{\partial}\lb(\xi^{n-k-2}\cdot\xi^P\rb)\ \vec{\partial}\xi^Q\partial_Q\lb(\xi^{k}\cdot \xi_P\rb)\cdot
			\end{eqnarray}
			Now we add the last two results along with the term $\lb(n+1\rb)\lb(\Delta_{n+1}-\Delta_{n+1}^t\rb)$ to obtain the $LHS$ of equation (\ref{Iter2}). After a straightforward relabeling of the summation index $k$ we obtain
			\begin{eqnarray}
				LHS\ =\ &&\SUL{k=1}{k=n-1}\lb(\lb(n+1\rb)\alpha_{n+1,k}+\alpha_{n,k}+\alpha_{n,k-1}\rb)\lb(\vec{\partial}\lb(\xi^{k}\cdot \xi^P\rb)\vec{\partial}\lb(\xi^{n-k-1}\cdot\xi_P\rb)-\lb(\cdots\rb)^t\rb)
				\nonumber\\ && +\lb(\lb(n+1\rb)\alpha_{n+1,0}+\alpha_{n,0}+\alpha_{n,n-1}\rb)\lb(\vec{\partial}\lb(\xi^P\rb)\vec{\partial}\lb(\xi^{n-1}\cdot\xi_P\rb)-\lb(\cdots\rb)^t\rb)\cdot
			\end{eqnarray}
			The $RHS$ of equation (\ref{Iter2}) can be written by using the definition of $A_n$
			\begin{eqnarray}
				RHS\ =\ \frac{1}{2}\frac{\left(-1\right)^{n}}{n!}\lb(\vec{\partial}\lb(\xi^P\rb)\vec{\partial}\lb(\xi^{n-1}\cdot\xi_P\rb)-\lb(\cdots\rb)^t\rb)\cdot
			\end{eqnarray}
			Let us write $\alpha_{n,k}$ as product of two factors \be \alpha_{n,k}=\frac{1}{2}\frac{\lb(-1\rb)^{n}}{n!} \times \beta_{n,k}\ee  where $\beta_{n,k}$ is some other $n$ and $k$ dependent function. For $LHS$ to be equal to the $RHS$ we get following conditions on $\beta_{n,k}$.
			\begin{eqnarray}
				-\beta_{n+1,0}+\beta_{n,0}+\beta_{n,n-1}\ &=&\ 1, \\
				-\beta_{n+1,k}+\beta_{n,k}+\beta_{n,k-1}\ &=& 0\cdot
			\end{eqnarray}
			Let us focus on the second condition first and write it as $\beta_{n+1,k}=\beta_{n,k}+\beta_{n,k-1}$, this is just the `Pascal's rule' in combinatorics if we choose $\beta_{n,k}={n-1 \choose k}$. Now for the first condition
			\be 
			-{n\choose 0} + {n-1\choose 0} +{n-1\choose n-1}=-1+1+1=1,
			\ee 
			so it is also satisfied. In summary, we deduce that 
			\be 
			\Delta_{n}-\Delta_{n}^t\ =\ \SUL{k=0}{k=n-2} \alpha_{n,k}\lb(\  \vec{\partial}\lb(\xi^{k}\cdot\xi^P\rb)\ \vec{\partial}\lb(\xi^{n-k-2}\cdot \xi_P\rb)\ -\ \lb(\cdots\rb)^t \rb) \text{ with } \alpha_{n,k}=\frac{1}{2}\frac{\lb(-1\rb)^{n}}{n!}{n-1 \choose k}\cdot\label{ddtfinal}
			\ee
			satisfy equation (\ref{Iter2}) for any $n$. One can check that this formula agrees with the results obtained in the last section, up to quintic order. From this formula, we can make a choice of $\Delta_{n}$ and hence $\Delta =\SUL{n}{\ }\Delta_{n}$ can be determined. From $\Delta$, one can then find $\chi\lb(\xi\rb)$. Using the procedure we described in the last section, one particular choice of $\chi\lb(\xi\rb)$ is
			\be 
			\chi^M\lb(\xi\rb)\ =\ \SUL{k=0}{k=n-2} \alpha_{n,k}  \ \lb(\xi^{k}\cdot\xi^P\rb)\ \partial^M\ \lb(\xi^{n-k-2}\cdot \xi_P\rb)\cdot\label{chiall}
			\ee
			Once we have $\chi\lb(\xi\rb)$, we can use (\ref{NewParDef}) to obtain $\xi'\lb(\xi\rb)$ to any order in $\xi$.
			\par In \cite{GlobalAspects} Berman et al. considered the HZ conjecture in the context of equivalence classes of finite transformations. Equivalence classes were defined modulo the `non-translating' transformations. These are the transformations which differ from the identity by a nilpotent matrix. Let $\mathcal{F}\lb(\xi\rb)$ be a representative of this equivalence class then it can be formally defined as follows:
			\be 
			\lb[\mathcal{F}\lb(\xi\rb)\rb]\ =\ \lb\{ \mathcal{F}\lb(\xi\rb)\lb( \mathbf{1} \ + \ Q\rb), \text{ where }\  Q^2\ =\ 0\rb\}.
			\ee 
			Then it was shown that  `$\mathcal{G}^{-1}(\xi)\ \mathcal{F}(\xi)$'  belongs to the same equivalence class as the identity. An expression for `$\mathcal{G}^{-1}(\xi)\ \mathcal{F}(\xi)$' was computed explicitly which can be written in the following form
			\be 
			\mathcal{G}^{-1}(\xi)\ \mathcal{F}(\xi)\ =\ \prod_{n=2}^{\infty}\prod_{k=0}^{k=n-1} \lb(\ \mathbf{1}\ +\ \kappa_{n,k}\ M_{n,k}\rb),\label{bermanResult}
			\ee
			where $\kappa_{n,k}$ is an $n$ and $k$ dependent coefficient and $M_{n,k}$ is a nil-potent matrix given by:
			\be  
			\kappa_{n,k}\ =\ \frac{1}{2}\frac{\lb(-1\rb)^n\lb(n-2k-1\rb)}{\lb(n+1\rb)\lb(k+1\rb)! \lb(n-k\rb)!} \ \ \text{and}\ \ \ M_{n,k}\ =\ \vec{\partial}\lb(\xi^k\cdot \xi^P\rb)\vec{\partial}\lb(\xi^{n-k-1}\cdot \xi_P\rb). \label{bermDef}
			\ee

			We want to compare the result (\ref{bermanResult}) with our analysis here. Since we have shown that $\mathcal{F}\lb(\xi\rb)\ =\ \mathcal{G}\lb(\xi'\lb(\xi\rb)\rb)\ =\mathcal{G}\lb(\xi\rb)\ \mathcal{G}\lb(\chi\lb(\xi\rb)\rb)$, we expect the RHS of (\ref{bermanResult}) to agree with our expression for $\mathcal{G}\lb(\chi\lb(\xi\rb)\rb)$. In the following we will show that this is indeed the case.
			\par We use the definition of $\mathcal{G}\lb(\chi\lb(\xi\rb)\rb)\ =\ \mathbf{1}\ +\Delta-\Delta^t$ and our result in equation (\ref{ddtfinal}). We also recall that $\Delta_n-\Delta_{n}^t=0$ for $n\leq 2$. Using these facts we can write
			\begin{eqnarray}
				\mathcal{G}\lb(\chi\lb(\xi\rb)\rb)\ &=&\ \mathbf{1}\ +\ \sum_{n=3}^{\infty}\ \sum_{k=0}^{n-2}\ \alpha_{n,k}\lb(\  \vec{\partial}\lb(\xi^{k}\cdot\xi^P\rb)\ \vec{\partial}\lb(\xi^{n-k-2}\cdot \xi_P\rb)\ -\ \lb(\cdots\rb)^t \rb).
			\end{eqnarray}
			The two terms in the sum can be combined by replacing the dummy index $k$ by $n-k-2$.
			\begin{eqnarray}
				\mathcal{G}\lb(\chi\lb(\xi\rb)\rb)\ &=&\ \mathbf{1}\ +\ \sum_{n=3}^{\infty}\ \sum_{k=0}^{n-2}\ \lb(\alpha_{n,k}\ -\ \alpha_{n,n-k-2}\rb)\lb(\  \vec{\partial}\lb(\xi^{k}\cdot\xi^P\rb)\ \vec{\partial}\lb(\xi^{n-k-2}\cdot \xi_P\rb)\  \rb), \\ 
				&=&\ \mathbf{1}\ +\ \sum_{n=3}^{\infty}\ \sum_{k=0}^{n-2}\ \lb(\alpha_{n,k}\ -\ \alpha_{n,n-k-2}\rb) M_{n-1,k},
			\end{eqnarray}
			where we have used the definition of $M_{n,k}$ from equation (\ref{bermDef}). We now replace the dummy index $n$ by $n+1$ to get
			\begin{eqnarray}
				\mathcal{G}\lb(\chi\lb(\xi\rb)\rb)\ &=&\ \mathbf{1}\ +\ \sum_{n=2}^{\infty}\ \sum_{k=0}^{n-1}\ \lb(\alpha_{n+1,k}\ -\ \alpha_{n+1,n-k-1}\rb) M_{n,k} . 
			\end{eqnarray}
			Since the $M_{n,k}$ are nil-potent matrices, we can write
			\begin{eqnarray}
				\mathcal{G}\lb(\chi\lb(\xi\rb)\rb)\ &=&\ \prod_{n=2}^{\infty}\ \prod_{k=0}^{n-1}\ \bigg(  \mathbf{1}\ +\ \lb(\alpha_{n+1,k}\ -\ \alpha_{n+1,n-k-1}\rb)M_{n,k}\bigg), \label{gchi}
			\end{eqnarray}
			From equation (\ref{ddtfinal}) we use the expression for $\alpha_{n,k}$ to compute
			\begin{eqnarray}
				\alpha_{n+1,k}-\alpha_{n+1,n-k-1}\ &=&\ \frac{1}{2}\frac{\lb(-1\rb)^{n+1}}{\lb(n+1\rb)!}\lb({n \choose k}-{n \choose n-k-1}\rb), \\
				&=&\ \frac{1}{2}\frac{\lb(-1\rb)^{n+1}}{\lb(n+1\rb)}\frac{1}{k!\ \lb(n-k-1\rb)!}\ \lb(\frac{1}{n-k}\ -\frac{1}{k+1} \rb), \\
				&=\ & \frac{1}{2}\frac{\lb(-1\rb)^n}{\lb(n+1\rb)}\frac{n-2k-1}{\lb(k+1\rb)!\lb(n-k\rb)!}, \\
				&=&\ \kappa_{n,k}.
			\end{eqnarray}
			Using this in equation (\ref{gchi}), we see that 
			\begin{eqnarray}
				\mathcal{G}\lb(\chi\lb(\xi\rb)\rb)\ &=&\ \prod_{n=2}^{\infty}\ \prod_{k=0}^{n-1}\ \bigg(  \mathbf{1}\ +\ \kappa_{n,k}\ M_{n,k}\bigg).
			\end{eqnarray}
			which is precisely the RHS of equation (\ref{bermanResult}) and we deduce that our result for $\mathcal{G}\lb(\chi\lb(\xi\rb)\rb)$ is in agreement with \cite{GlobalAspects} . However the relationship between $\chi$ and $\xi'$ was not discussed in \cite{GlobalAspects}. We made explicit the connection between $\chi(\xi)$ and $\xi'(\xi)$ due to our results in the section (\ref{decG}), which play a crucial rule in proving Hohm and Zwiebach's conjecture.
			\section{Conclusions and Outlook}\label{concl}
			We conclude this paper by summarizing our results and discussing the important issue of composition of  finite transformations in double field theory. We will also comment on the finite transformations in exceptional field theory.
					
			We have shown that the formula for finite gauge transformations in double field theory is equivalent to the transformations obtained by exponentiating the generalized Lie derivative. In particular we proved that for every parameter $\xi^M$ we can find a parameter $\xi'^M\lb(\xi\rb)$ such that \be \mathcal{F}\lb(\xi\rb)\ =\ \mathcal{G}\lb(\xi'\lb(\xi\rb)\rb),\ee 
			and we showed that $\mathcal{G}\lb(\xi'\lb(\xi\rb)\rb)$ can be decomposed as $\mathcal{G}\lb(\xi\rb)\mathcal{G}\lb(\chi\lb(\xi\rb)\rb)$, with $\chi\lb(\xi\rb)$ being a quasi-trivial parameter. 
			We have also given an iterative procedure which can be used to determine $\xi'^M\lb(\xi\rb)$ to any order in $\xi$. We used this procedure to obtain results up to quintic order and verified them explicitly.
			In the last section we also gave an explicit formula for $\chi\lb(\xi\rb)$ to all orders in $\xi$. We also showed that $\xi'\lb(\xi\rb)$ obtained using this procedure is not unique. This non-uniqueness can be understood in terms of the inherent non uniqueness of the generalized Lie derivative because the generalized Lie derivative with respect to a trivial parameter is zero.
					
			We note that the composition rules for $\mathcal{G}\lb(\xi\rb)$ are well understood. After proving (\ref{claim}), we are now in a better position to understand the composition of $\mathcal{F}\lb(\xi\rb)$. Composition of $\mathcal{G}\lb(\xi\rb)$ was studied in section 5.2 of \cite{LargeGT} and it was shown that
			\be 
			\mathcal{G}\lb(\xi_2\rb)\mathcal{G}\lb(\xi_1\rb)\ =\ \mathcal{G}\lb(\xi^C\lb(\xi_2,\xi_1\rb)\rb),
			\ee
			with 
			\be 
			\xi^C\lb(\xi_2,\xi_1\rb) =\ \xi_2+\xi_1+\frac{1}{2}\lb[\xi_2,\xi_1\rb]_C  +\cdots\ ,\label{xi12c}
			\ee
			where `$\cdots$' contains further nested C-brackets between $\xi_2$ and $\xi_1$ given by the BCH formula. The issue of composition was also discussed in \cite{LargeGT}, where generalized coordinate transformations generated by a parameter $\Theta\lb(\xi\rb)$  were considered such that $\mathcal{F}\lb(\Theta\lb(\xi\rb)\rb)=\mathcal{G}\lb(\xi\rb)$. In our language, this means that $\xi'\lb(\Theta\lb(\xi\rb)\rb)\ =\ \xi$. So under the coordinate transformations, 
			\begin{eqnarray}
				X'\ =\ e^{-\Theta\lb(\xi_1\rb)}\  X,\ \ \ \ \ \ 
				X''\ =\ e^{-\Theta\lb(\xi_2\rb)} \ X',
			\end{eqnarray}
			the composition of $\mathcal{F}$ is given by:
			\begin{eqnarray}
				\mathcal{F}\lb(\Theta\lb(\xi_1\rb)\rb)\mathcal{F}\lb(\Theta\lb(\xi_2\rb)\rb)\ &=&\ \mathcal{G}\lb(\xi_2\rb)\mathcal{G}\lb(\xi_1\rb), \\
				&=&\ \mathcal{G}\lb(\xi^C\lb(\xi_2,\xi_1\rb)\rb),\\
				&=&\ \mathcal{F}\lb(\Theta\lb(\xi^C\lb(\xi_2,\xi_1\rb)\rb)\rb),
			\end{eqnarray} which implies that the composition of coordinates is given by
			\be 
			X''\ =\ e^{-\Theta\lb(\xi^C\lb(\xi_2,\xi_1\rb)\rb)}\  X.
			\ee 
			We want to understand this composition more directly, i.e., in terms of generalized coordinate transformations given by
			\begin{eqnarray}
				X'\ =\ e^{-\xi_1} X,\ \ \ \ \ \ 
				X'\ =\ e^{-\xi_2} X'.
			\end{eqnarray}
			The corresponding $\mathcal{F}$ matrices are given by:
			\begin{eqnarray}
				\mathcal{F}\lb(\xi_1\rb)\ =\ \mathcal{G}\lb(\xi_1+\delta\lb(\xi_1\rb)\rb)\ =\ \mathcal{G}\lb(\xi_1\rb)\mathcal{G}\lb(\chi\lb(\xi_{1}\rb)\rb)\\
				\mathcal{F}\lb(\xi_2\rb)\ =\ \mathcal{G}\lb(\xi_2+\delta\lb(\xi_2\rb)\rb)\ =\ \mathcal{G}\lb(\xi_2\rb)\mathcal{G}\lb(\chi\lb(\xi_2\rb)\rb)
			\end{eqnarray}
			where $\delta\lb(\xi_i\rb)$ and $\chi\lb(\xi_i\rb)$ are determined using equations (\ref{NewParDef}) and (\ref{chiall}). It is now easy to see that
			\begin{eqnarray}
				\mathcal{F}\lb(\xi_2\rb)\mathcal{F}\lb(\xi_1\rb)\ &=&\ \mathcal{G}\lb(\xi_2+\delta\lb(\xi_2\rb)\rb)\mathcal{G}\lb(\xi_1+\delta\lb(\xi_1\rb)\rb)\\
				&=&\ \mathcal{G}\lb(\xi^C\lb(\xi_2+\delta\lb(\xi_2\rb),\xi_1+\delta\lb(\xi_1\rb)\rb)\rb)\\
				&=&\ \mathcal{F}\lb(\Theta\lb(\xi^C\lb(\xi_2+\delta\lb(\xi_2\rb),\xi_1+\delta\lb(\xi_1\rb)\rb)\rb)\rb)
				.\label{composition}
			\end{eqnarray}
			This composition rule then implies that the composition of the generalized coordinate transformations is given by
			\be 
			X''\ =\ e^{-\Theta\lb(\xi^C\lb(\xi_2+\delta\lb(\xi_2\rb),\xi_1+\delta\lb(\xi_1\rb)\rb)\rb)}\ X . \label{CompResult}
			\ee
			We conclude the discussion about the composition by noting that the composition rule (\ref{CompResult}) is very unconventional and it has important consequences regarding associativity  as discussed in section 6 of \cite{LargeGT}. 
					
			In the context of exceptional field theory (see refs. \cite{ExFormSugra,GenGeoM,GaugeGenDiff,MtheoEgeo,E7nSUGRA,DInvaAction,SUGRAasGenGeo,SL5nU,ExGeonTensFields,NonGravESup,EFT1,EFT2,EFT3,GCartanCal,EFTSO5,ExGeo}), the issue of finite transformations still needs to be understood. In particular there is no analogue for the matrix $\mathcal{F}$ which gives finite transformations of the fields. We conclude this paper with some remarks about this issue.  We will follow the notation of \cite{ExGeo} closely. In exceptional field theory, the generalized Lie derivative takes the following form
			\be 
			\gld{\xi}=\xi+a-a^Y,
			\ee
			where $\lb(a^Y\rb)_{M}^{\ \ N}\ \equiv Y_{MP}^{\ \ \ QN} a_{Q}^{\ \ P}$, $a_{M}^{\ \ N}=\partial_M\xi^N$ and $Y_{MP}^{\ \ \ QN}$ is the $E_{n(n)}$ invariant tensor.  The notion of C-bracket is replaced by a generalized Lie bracket called `E-bracket' defined as follows:
			\be 
			\lb[\xi_1,\xi_2\rb]_{E}^{M}\ =\ \xi_{1}^{P}\partial_P \xi_{2}^{M}-\xi_{2}^{P}\partial_P \xi_{1}^{M} - \frac{1}{2}Y_{PQ}^{\ \ \ MN}\lb(\xi_{1}^{Q}\partial_N\xi_{2}^{P}-\xi_{2}^{Q}\partial_N\xi_{1}^{P}\rb)
			.\ee
			Similar to the case of double field theory, generalized Lie derivatives form a Lie algebra given by
			\be 
			\lb[\gld{\xi_1},\gld{\xi_2}\rb]\ =\ \gld{\lb[\xi_1,\xi_2\rb]_E}.
			\ee
			The closure of this Lie algebra puts some constraints on $Y_{PQ}^{\ \ \ MN}$. These constraints include an analogue of the strong constraint in double field theory (see section 6.4 of \cite{ExGeo} for a complete list of constraints)
			\be 
			Y_{MP}^{\ \ \ QN}\partial_Q\partial_N\lb(\cdots\rb)=0\label{newSC},
			\ee
			where `$\cdots$' is any product of fields and parameters. Now we modify the definition of the quasi-trivial parameter and for the exceptional case, we define a parameter of the following form to be quasi-trivial.
			\be  V^N=\sum_{i} Y_{MP}^{\ \ \ QN}\rho^{M}_i\partial_Q\eta^{P}_i.\label{newqt}\ee  It is easy to see using equation (\ref{newSC}), that  $V^P\partial_P$ gives zero when acting on product of fields and parameters. The analogue of the matrix $\mathcal{G}\lb(\xi\rb)$ is defined as:
			\begin{eqnarray}
				\mathcal{G}_E\lb(\xi\rb)\ =\ e^{-\xi}e^{\xi+a-a^Y}.\label{Gexc}
			\end{eqnarray}
			The relevant similarities between double field theory and exceptional field theory end here. Now note that we do not have analogue of equation (\ref{Gqt2}). This follows due to the fact that for a quasi-trivial parameter $\xi^M$, the matrix $a-a^Y$ is not nilpotent in general. This, in turn, follows from the fact that the E-bracket(as opposed to the C-bracket) of two quasi-trivial parameters is not zero in general. Also we do not have the analogue $\mathcal{F}_E$ of the matrix $\mathcal{F}$ which implements the finite transformations of the fields and can be written in terms of the generalized coordinate transformations (i.e., in terms of $\frac{\partial X}{\partial X'}$ and $\frac{\partial X'}{\partial X}$).
					
			While finding the matrix $\mathcal{F}_E$  remains a challenging open question we argue that regardless of the final form of $\mathcal{F}_E$, we should be able to express it in the form $ \mathcal{F}_E\lb(\xi\rb)=\mathcal{G}_{E}\lb(\xi+\gamma\lb(\xi\rb)\rb)$, where $\gamma(\xi)$ is quasi-trivial parameter and $X'=e^{-\xi}X$. This conclusion follows by considering the transformation of a scalar field and using similar kind of arguments as in section 2.1. Our analysis in section \ref{decG} can also be generalized\footnote{For the exceptional case, equation \ref{lem1} does not hold but it does not present any difficulty in generalizing the analysis. It only leads to a different expression for $\chi$ in terms of $\delta$ but the quasi-trivial nature of $\chi$ and $\delta$ is preserved.} to the case of exceptional field theory and we can deduce that $\mathcal{G}_{E}\lb(\xi+\gamma\lb(\xi\rb)\rb)=\mathcal{G}_E\lb(\xi\rb)\mathcal{G}_E\lb(\Gamma\lb(\xi\rb)\rb)$ where $\gamma$ and $\Gamma$ satisfy a similar relation as in equation (\ref{NewParDef}) (with $\delta$, $\chi$ and C-brackets replaced by $\gamma$, $\Gamma$ and E-brackets respectively). So we deduce that $\mathcal{F}_E\lb(\xi\rb)=\mathcal{G}_E\lb(\xi\rb)\mathcal{G}_E\lb(\Gamma\lb(\xi\rb)\rb)$. Now the problem of finding $\mathcal{F}_E$ can be stated as: Is there an expression for $\Gamma\lb(\xi\rb)$ such that the RHS of the last equation can be written in terms of $\frac{\partial X}{\partial X'}$ and $\frac{\partial X'}{\partial X}$ where $X'=e^{-\xi}X$? We stress that this reformulation of the problem does not make it any easier as it is not immediately obvious how one should approach this new problem. However it does provide a different way to think about the problem of finding $\mathcal{F}_E$.
			\acknowledgments
			I would like to thank Barton Zwiebach for numerous discussions, very useful comments and kind supervision throughout the course of this work. This work is supported by the U.S. Department of Energy under grant Contract Number DE-SC00012567.

	\end{document}